%% file: main_v1.tex
\documentclass{article}

\usepackage{PRIMEarxiv}

\usepackage{hyperref}       
\usepackage{url}            
\usepackage{booktabs}       
\usepackage{amsfonts}       
\usepackage{nicefrac}       
\usepackage{microtype}      
\usepackage{lipsum}
\usepackage{fancyhdr}       
\usepackage[english]{babel}
\usepackage{amsmath,amssymb,amsthm}
\newtheorem{theorem}{Theorem}[section]
\newtheorem{definition}{Definition}[section]

\usepackage{enumitem}

\usepackage{pgfplots}
\pgfplotsset{compat = newest}

\usepackage{graphicx}       
\graphicspath{{media/}}     
\usepackage{subcaption}
\usepackage{wrapfig}
\usepackage{float}
\usepackage[skip=3pt, indent=15pt]{parskip}

\usepackage{booktabs}
\usepackage{array}
\newcolumntype{L}[1]{>{\raggedright\arraybackslash}p{#1}}

\usepackage{todonotes}

\input{preamble}

\title{Robust Risk-Aware Option Hedging
}

\author{David Wu 
\\Operations Research and Information Engineering \\ Cornell University 
\\ 
\href{mailto:dbw78@cornell.edu}{\texttt{dbw78@cornell.edu}} 
\and
\textbf{Sebastian Jaimungal}
\\Department of Statistical Sciences \\University of Toronto \\
\& \\
Oxford-Man Institute \\
University of Oxford \\
\href{mailto:sebastian.jaimungal@utoronto.ca}{\texttt{sebastian.jaimungal@utoronto.ca}}
}

\begin{document}
\maketitle

\begin{abstract}

The objectives of option hedging/trading extend beyond mere protection against downside risks, with a desire to seek gains also driving agent's strategies. In this study, we showcase the potential of robust risk-aware reinforcement learning (RL) in mitigating the risks associated with path-dependent financial derivatives. We accomplish this by leveraging a policy gradient approach that optimises robust risk-aware performance criteria. We specifically apply this methodology to the hedging of barrier options, and highlight how the optimal hedging strategy undergoes distortions as the agent moves from being risk-averse to risk-seeking. As well as how the agent robustifies their strategy. We further investigate the performance of the hedge when the data generating process (DGP) varies from the training DGP, and demonstrate that the robust strategies outperform the non-robust ones.
\end{abstract}

\section{Introduction}

Option hedging is a foundational issue in the field of mathematical finance, tracing back to the groundbreaking work of Black and Scholes \cite{black1973pricing} and Merton \cite{merton1973tro}. While exact replication of options is achievable in frictionless and complete markets, real-world markets entail transaction costs and necessitate discrete-time trading, rendering perfect replication unattainable. Notably, even under continuous trading conditions, delta hedging ceases to be optimal when trading costs are considered, and it is impossible to replicate arbitrary payoffs. Various studies have addressed the influence of market frictions on option hedging, including those by  Whalley and Wilmott \cite{Whalley1997} and Martinelli \cite{Martinelli2000}.

Contemporary research has increasingly focused on developing more robust, model-agnostic approaches to option hedging, with a growing interest in machine learning techniques (see Ruf and Wang \cite{Ruf2020ANNreview} for an overview). For instance, Cao et al. \cite{Cao2021rotman} and Kolm and Ritter \cite{kolm2019} have proposed hedging strategies that involve solving mean-variance optimization problems using neural networks, Lütkebohmert, et. al. \cite{lutkebohmert2022robust} study pricing under parameter uncertainty, and Gierjatowicz, et.al. \cite{gierjatowicz2020robust}  find robust bounds for prices of derivatives using neural stochastic differential equations. Recently\footnote{In parcticular, \cite{limmer2023robust} appeared online after the current work was submitted for review.}, Limmer \& Horvatz \cite{limmer2023robust} study robustification by selecting sample paths from a bank of models and penalise models by the amount they deviate from a reference model. 

As perfect replication is unattainable, agents  must bear some risk exposure when buying/selling options. The  agent's risk tolerance and appetite dictate the profit and loss (P\&L) profile they are willing to accommodate. Rather than aiming to replicate option payoffs, our approach, drawing upon the work of Ilhan et.al. \cite{ilhan2009} and Buehler et.al. \cite{buehler2019deep}, seeks to minimise measures of risk.   We propose employing rank dependent expected utility (RDEU) to more comprehensively represent agents' risk preferences, as suggested by Yaari \cite{yarri1987}. RDEU enables agents to incorporate the utility associated with outcomes and account for distortions in the probability of those outcomes, thereby facilitating both risk-averse and gain-seeking behaviour concurrently.

Existing literature on option hedging using machine learning largely neglects model misspecification and potential market condition changes. To safeguard agents against such uncertainties, we capitalise on advances in robust reinforcement learning (RL) as outlined by Rahimian and Mehrotra \cite{1908.05659}. Specifically, we employ the framework presented by Jaimungal et al. \cite{jaimungal2022robust} to robustify RDEU-based strategies and apply them to the option hedging problem.

In this framework, agents assess a strategy not by the RDEU of terminal wealth, but by distorting its distribution within a predefined ambiguity set. The ambiguity set comprises all distributions of terminal wealth situated within a Wasserstein ball surrounding the terminal wealth. Agents subsequently select actions that are optimal under the worst-case scenario within this ambiguity set. Given the absence of additional assumptions, this approach offers protection against various forms of model uncertainty, including those stemming from underlying asset process misspecification or alterations in market frictions. A similar approach for infinitesimally small uncertainty sets was investigated by Bartl, et.al. \cite{bartl2021sensitivity} where they consider sensitivity of a generic stochastic optimisation problem to model uncertainty -- in a perturbative fashion.

To approximate the agents' strategies, we employ neural networks. At each trading time, the neural network receives features representing the state and produces the hedge and we take a policy gradient approach to improve the policy.

We apply our algorithm to hedge barrier options, demonstrating its capacity to generate sound hedges for stochastic volatility models such as Heston, in the presence of market frictions. Importantly, our algorithm is not contingent upon a specific model and can be applied to any price process or derivative. Moreover, we illustrate that our robust models surpass the classic complete-market outcomes provided by Black-Scholes formulas for pricing barrier options.

The remainder of this article is organised as follows.  We provide the relevant mathematical background such as RDEU in Section \ref{sec: RM} and  Wasserstein distances in Section \ref{sec:Wasserstein}. Section \ref{sec: problem} formalises the robust and non-robust option hedging problems. Section \ref{sec:algorithm} provides policy gradient formulae for solving the optimisation problems. We  present results for hedging barrier call options in Section \ref{sec:results}, where we also investigate how robustness changes optimal strategies. Further examples demonstrate that robust hedges outperform non-robust hedges under model misspecification. 

\section{Risk Measures}

This section provides the basic definitions of the risk measures that we study and their key representation in terms of the distortion and quantile functions.
\label{sec: RM} We work on a probability space $(\Omega, \F, \P)$, moreover, let $\X$ represent the set of $\F$-measurable, $\P$-integrable, random variables. Let $Z_1, Z_2 \in \X$ be interpreted as random costs imposed on the agent. 

\begin{definition}
A {\normalfont risk measure} is a mapping $\rho: \mathcal{X} \to \mathbb{R}$ and is said to be
\begin{enumerate}[noitemsep]
    \item {\normalfont Monotone} if $Z_1 \leq Z_2$ implies $\rho(Z_1) \leq \rho(Z_2)$;
    \item {\normalfont Translation Invariant} if for $m \in \mathbb{R}, \rho (Z_1 + m) = \rho(Z_1) + m$;
    \item {\normalfont Positive Homogeneous} if for $\beta > 0, \rho(\beta Z_1) = \beta \rho(Z_1)$;
    \item {\normalfont Subadditive} if $\rho(Z_1 + Z_2) \leq \rho(Z_1) + \rho(Z_2)$;
    \item {\normalfont Coherent} if and only if conditions 1 -- 4 are satisfied
\end{enumerate}
\end{definition}

In this work, we focus on measuring risk using RDEU risk measures \cite{yarri1987}. These measures fall outside the class of coherent risk measures \cite{artzner1999coherent} but also include many special cases such as distortion risk measures and expected utility. Agents value wealth according to a utility function, but also consider probabilities of the events subjectively through a distortion function. Formally, RDEU is defined as follows.
\begin{definition}
Let g be an increasing function with g(0) = 0 and g(1) = 1. Let U be a non-decreasing concave an differentiable almost everywhere utility function. 
The RDEU of a random variable Z is defined as
$$
\R^U_g[Z] := \int_{-\infty}^{0} \Big(1 - g[\mathbb{P}(U(Z) > z)] \Big)\,dz - \int_{0}^{+\infty} g[\mathbb{P}(U(Z) > z)]\,dz
$$

\end{definition}
Under suitable conditions, $R^U_g[Z]$ has a more convenient representation as shown in Theorem \ref{thm: Alternate RDEU}.
\begin{theorem}[{\normalfont Representation of RDEU}]
\label{thm: Alternate RDEU}
If the distortion function g is left-differentiable, the RDEU can be rewritten as
$$\R^U_g[Z] := - \int_{0}^{1} U(F^{-1}_Z(s)) \,\gamma(s) \,ds$$
where $\gamma$ is given by $\gamma(u) := \partial_- g(x)|_{x= 1 - u}$, $\partial_-$ denotes the left hand derivative, and $F_Z^{-1}(s)$ denotes the quantile function of the random variable $Z$.
\end{theorem}

By utilising the agent's utility function and allowing for distortions in probabilities, the RDEU framework enables us to consider the agent's risk/gain seeking behaviour and ultimately provide a more comprehensive representation of their risk preferences. As option payoffs are not replicable outside of complete market settings, we utilise risk measures to quantify the agent's preferences when hedging.  

\subsection{Risk avoiding while Profit Seeking: \texorpdfstring{$\alpha-\beta$}{a-b} risk measures}
\label{sec: special}

In this subsection, we provide a prototypical example of an RDEU risk measure that we utilise in the hedging constructions. We employ  them  as they allow us to showcase how the investor reacts to avoiding risk while seeking gains with only a few parameters. Henceforth, we call such risk measures simply as $\alpha-\beta$ risk measures.
\begin{definition}[$\alpha-\beta$ risk measure (\cite{jaimungal2022robust})]
Let
$$\gamma^{\alpha,\beta}_p (u)= \frac{1}{\eta} \Big[\, p\, \mathbb{I}(u \leq \alpha) + (1-p)\, \mathbb{I}(u \geq \beta) \,\Big]$$
for $0< \alpha \leq \beta < 1$ and $p \in [0,1]$. $\eta := p\,\alpha + (1-p)(1-\beta)$ is  a normalising constant s.t. $\int_{0}^{1} \gamma^{\alpha, \beta}_p (u) = 1$. The $\alpha-\beta$ risk measure of a r.v. $Z\in\X$ is defined as:
\begin{equation}
    R^{\alpha,\beta}_p[Z] := \R_{\gamma^{\alpha,\beta}_p}^{f}[Z],
\end{equation}
where $f$ is the identify function, i.e., $f(x)=x$.
\end{definition}

Figure \ref{fig:gamma} provides a visual representation of $\gamma(u)$ for select parameters. Different values of $\alpha$,  $\beta$, and $p$ result in U-shaped distortion functions, which model agents who are concerned about profits below the $\alpha$ quantile and above the $\beta$ quantile. Moreover, if $p > \frac{1}{2}$, the investor emphasises losses over gains. The opposite is true for $p < \frac{1}{2}$. 
\begin{figure}[h]
  \centering
  \includegraphics[width=0.8\textwidth]{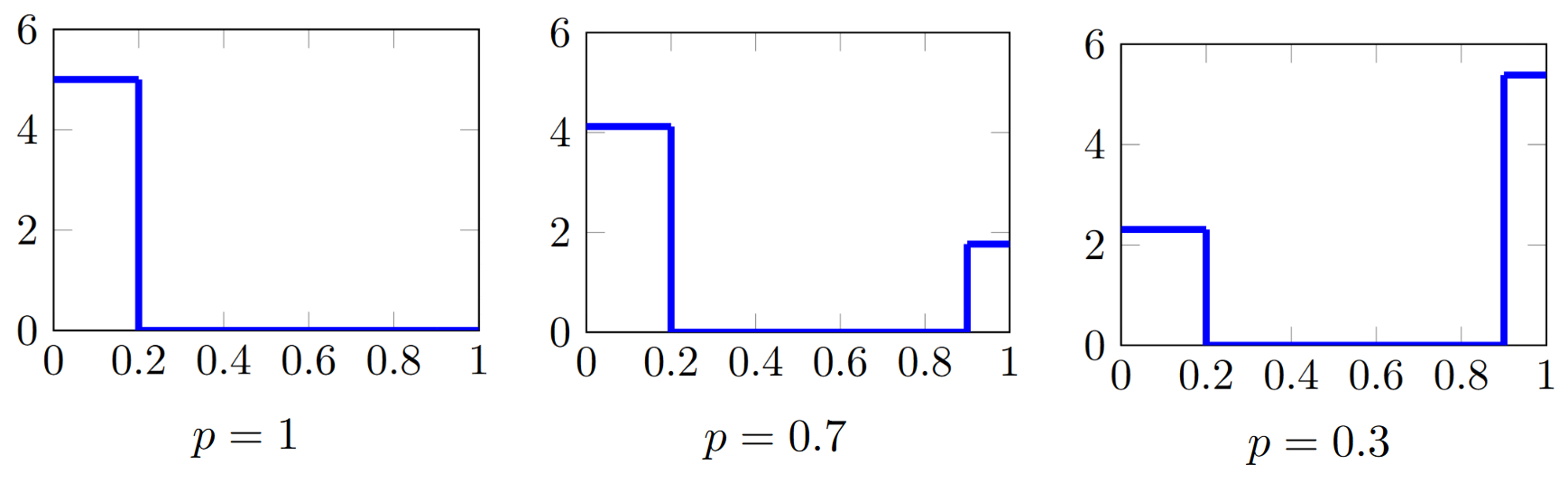}
  \caption{Depictions of $\gamma(u)$ for $\alpha = 0.2$ and $\beta = 0.9$ as $p$ varies. }
  \label{fig:gamma}
\end{figure}

The $\alpha-\beta$ risk measures contain several notable special cases, such as expected utility and Conditional-Value-at-Risk (CVaR). For example, when $p = 1$, as shown on the leftmost panel of Figure \ref{fig:gamma}, the RDEU reduces to  CVaR at level  $\alpha$. Furthermore, such risk measures are translation invariant and positive homogeneous. These properties play a role when analysing the optimal hedging strategies. 

\section{Wasserstein Distances}
\label{sec:Wasserstein}
In this section, we present a brief overview of Wasserstein distances, which help establish a metric on the set of probability distributions and allow us to define our ambiguity set. For a more detailed coverage of the topic, the reader is referred to, e.g., \cite{Santambrogio}. In many settings, and in particular for option hedging, the asset price process may be estimated using past data. Once estimated, we can select a strategy and determine its terminal P\&L distribution. A natural way to protect the agent from model uncertainty is to consider a family of distributions within a neighbourhood, almost like a confidence interval, of our estimated distribution. Wasserstein distances provide a way to do so. 
Let $\mathcal{P}(\Xsp)$ be the space of probability measures on $\Xsp \subset \mathbb{R}^d$, and define: 
$$\mathcal{P}_p(\Xsp) = \left \{ \mu \in \mathcal{P}(\Xsp)\;\left|\; \int_{\Xsp} | x | ^ p d \mu(x) < + \infty \right. \right \}$$
with $p\ge1$, i.e. the space of probability measures with bounded $p^{\normalfont th}$ moment. 

\begin{definition}
Given measures $\mu, \nu \in \mathcal{P}_p(\Xsp)$, the Wasserstein distance of order p is defined as
\begin{equation}
d_p [\mu, \nu] = \inf_{\pi \in \Pi(\mu, \nu)} \left( \int_{\Xsp \times \Xsp} \left| x - y\right|^p d\pi(x,y) \right) ^{\frac{1}{p}}\,,
\label{eqn:wasserstein}
\end{equation}
where $\Pi(\mu, \nu)$ denotes the set of \textit{transport plans} and contains all measures $\pi$ such that $\pi(A\times \Xsp) = \mu(A)$ and $\pi(\Xsp \times B) = \nu(B)$, for all measurable sets $A \subseteq \Xsp$ and $B \subseteq \Xsp$. 
\end{definition}
Intuitively, the Wasserstein distance measures the cheapest way to transport mass between distributions. The constraints on the transport plans ensures that the total mass removed from a measurable set $A \subseteq \Xsp$ is equal to $\mu(A)$, and total mass transferred to $B \subseteq \Xsp$ is equal to $\nu(B)$.

We focus on the single asset hedging problem which leads to a univariate P\&L random variable. In one-dimension, the optimal coupling that attains the $\inf$ in \eqref{eqn:wasserstein} is the comonotonic coupling. This  leads to a tractable representation of the Wasserstein distance. Consider two measures $\mu, \nu \in \mathcal{P}(\mathbb{R})$ having cumulative distribution functions (cdf) $F$ and $G$, respectively. Define the generalised inverse of a function $F$ on $[0,1]$ as $F^{-1}(t) = \inf\{x \in \mathbb{R} \mid F(x) >t\}$, i.e., when $F$ is a cdf then $F^{-1}$ is the quantile function. The Wasserstien distance is given by (see, e.g., \cite[Chap. 2]{ambrosio2003lecture})
$$
d_p[\mu, \nu] = \left( \int^{1}_{0} \left| F^{-1}(u) - G^{-1}(u) \right| ^p du \right).
$$

\section {Problem Formulation}
\label{sec: problem}
In this section, we provide an overview of the option hedging problem and how robust methods can be used to protect against model uncertainty. 

Consider an agent who shorts an option with payoff $V_T$ at maturity date. The price earned from selling the option at $t=0$ is denoted by $V_0$. We assume that the asset does not pay dividends, but this may be easily incorporated into the setup. Let $(S_t)_{t \in [0,T]}$  denote the underlying asset price process and -- with a slight abuse of notation --- let $(S_0, S_1, ... , S_{N-1}, S_N)$ represent the asset price at time points $0 = t_0 < t_1 < t_2 ... < t_{N-1} < T$. 
At each $t_i$, the agent observes the stock price and chooses the amount $\Delta_i$ of the asset to hold. For convenience, let $ (\Delta_i^\phi)_{i = 0, 1 ... N-1}$ represent the series of actions, and $\varphi$ represent the set of all hedging strategies. 

In some examples, we include market frictions in the form of proportional transaction costs. Formally, given a constant $c \in \mathbb{R}_+$, if the stock price is $S_t$ and $q$ shares are purchased, we impost a cost of $c\, S_t \,q$. Then, given a sequence of trades $\phi$ the agent's terminal wealth is given by
\begin{equation}
\label{eqn:wealth}
    X^\phi = V_0 +  \sum_{i=0}^{N-1} \Delta_i^\phi(S_{i+1} - S_i) - \sum_{i=0}^{N-1}c \left| \Delta_i^\phi - \Delta_{i-1}^\phi \right| S_i - V_{N},
\end{equation}
where $\Delta_{-1}^\phi := 0$ 
represents the asset holdings before any trading.

For simplicity, we set interest rates to 0, however, positive or even stochastic interest rates can be easily incorporated. Furthermore, as $\alpha-\beta$ risk measures  are translation invariant, we can safely ignore $V_0$ in \eqref{eqn:wealth} when optimizing $X^\phi$.

\subsection{Non-Robust Formulation}
\label{sec: non-rob}
We first provide the non-robust formulation of the problem. The agent seeks  to minimise the risk $\R^U_g[X^\phi]$ of their terminal wealth over the parameter set $\phi\in\varphi$. Assuming the agent is certain that the model is correctly specified (i.e. the distribution of $X^\phi$ is known with certainty for all actions $\phi$),  then the their optimisation problem may be stated as
\begin{equation}
    \inf_{\phi \in \varphi} \R^U_g[X^\phi],
\end{equation}
where the strategy $\phi$ is parametrised by a neural network that takes observed states as inputs and provides the strategy as output -- e.g., a natural choice of states would be the current time and asset price. Thus, the goal is to adjust the weights of the neural net to select the actions that induce the lowest risk. 

\subsection{Robust Formulation}
We also provide a robust formulation of the problem, as agents often do not have full knowledge of the model that drives the asset prices. Even if the agent uses historical data, they should robustify strategies to avoid obtaining strategies that are tuned to the historical price dynamics and work poorly out of sample. If agents use parametric models, then there is uncertainty in the value of those parameters, and even uncertainty whether they chose the correct parametric class of models. A third example of model uncertainty stems from financial markets not being stationary, hence a model that is well calibrated to some slice of historical data may not explain the dynamics well in the future.

As before, we use neural networks to learn the optimal mapping between observed states to $\phi$. To build robustness into the model, we consider all probability distributions in an \textit{ambiguity set}. Specifically, let $\vartheta \subseteq \mathbb{R}^n$ parameterise the ambiguity set, and $X^\theta$ be a $\mathbb{R}$-valued random variable parameterised by $\theta \in \vartheta$. When the agent considers model uncertainty, $X^\phi$ from the previous section serves as an approximation to the true wealth random variable. Thus, the agent builds uncertainty around this base model and uses an ambiguity set $\vartheta_\phi \subset \vartheta$ defined below. To protect against uncertainty, the agent picks actions to minimise the risk associated with the worst case distribution in $\vartheta_\phi$. Under this framework, we define the agent's optimisation as
\begin{equation}
    \inf_{\phi \in \varphi} \sup_{\theta \in \vartheta_\phi} \R^U_g \left[X^\theta \right] \textrm{\hspace{0.8cm} where \hspace{0.8cm} } \vartheta_\phi := \left\{ \theta \in \vartheta \mid d_p \left[X^\theta, X^\phi \right] \leq \epsilon \right\}.
\end{equation}
This formulation may be viewed as an adversarial attack. The agent picks actions $\phi\in\varphi$ which induces a terminal wealth of $X^\phi$. The adversary then distorts $X^\phi$ to $X^\theta$ that corresponds to the worst performance within a Wasserstein ball around $X^\phi$. To model this relationship, we set $X^\theta = H_\theta (X^\phi)$  where $H_\theta(\cdot)$ is an neural network parameterised by $\theta$. If  $X^\phi$ and $X^\theta$ are both continuous r.v., then, there always exists a neural network $\theta$ that achieves equality in distribution -- which is all that we need as the risk measures we employ are law invariant. Specifically, let  $F_{X^\theta}(x)$ denote the cdf of $X^\theta$, then the r.v. 
$
F_{X^\theta}^{-1}\big(F_{X^\phi}(X^\phi)\big)
$
has cdf $F_{X^\theta}(x)$ and, as under the assumption that $X^\theta$ and $X^\phi$ are continuous r.v., we have that  $F_{X^\theta}^{-1}\big(F_{X^\phi}(x)\big)$ is continuous and hence there exists a neural network that approximates it (on a finite domain) arbitrarily well.

This formulation allows for uncertainty on the r.v. $X^\phi$. Thus, it incorporates potentially both uncertainty in the underlying processes that generate $X^\phi$, transaction costs, and the strategy itself.

\section{Optimisation Algorithm}
\label{sec:algorithm}
In this section, we derive policy gradients for solving the non-robust and robust problems. The robust formulation reduces to the non-robust problem as $\epsilon \to 0$. However, solving the non-robust version is less computationally intensive, so we present separate algorithms for the two cases.

\subsection{Non-Robust Problem}
We optimise $\R^U_g[X^\phi]$ over the possible actions by using batch gradient descent where the gradient is taken with respect to the parameters of the policy. However, many risk measures such as CVaR yield a derivative with discontinuities. Na\"ive back-propagation leads to inaccuracies and we provide a different gradient formula.

\begin{theorem}
\label{thm: Non-Rob Grad}
Let $F_\phi(x)$ and $f_\phi(x)$ be the cdf and pdf of $X^\phi$, respectively.
Under the assumptions in Theorem \ref{thm: Alternate RDEU}, and further assuming that $X^\phi$ is a continuous r.v., we have that
$$
\nabla_\phi \R^U_g[X^\phi] = \E \left[ U'(X^\phi) \;\gamma\left(F_\phi(X^\phi) \right) \frac{\nabla_\phi F_\phi(x) \mid _{x = X^\phi}}{f_\phi(X^\phi)}\right]\,.
$$
\end{theorem}

\begin{proof}
As $F_\phi\left(F^{-1}_\phi(s)\right) = s$, taking gradients on both sides and using the chain rule we have
$$
f_\phi \left( F^{-1}_\phi(s)\right) \;\nabla_\phi F^{-1}_\phi(s)+
\nabla_\phi F_\phi\left(x\right)|_{x=F^{-1}_\phi(s)}
\nabla F^{-1}_\phi(s)
= 0
$$
and rearranging yields 
$$
\nabla_\phi F^{-1}_\phi(s) = - \left.\nabla_\phi F_\phi(x)\Big|_{x = F^{-1}_\phi(s)}\,\right/\,f_\phi\left(F^{-1}_\phi(s)\right).
$$
Using the representation of $\R^U_g[X^\phi]$ in Theorem \ref{thm: Alternate RDEU}, we have that
\begin{align}
\nabla_\phi \R^U_g[X^\phi] &= - \int_{0}^{1} U'(F^{-1}_\phi(s))\; \nabla_\phi F^{-1}_\phi(s) \;\gamma(s) ds \\
&=  \int_{0}^{1} U'(F^{-1}_\phi(s)) \; \gamma(s)\; \frac{\nabla_\phi F_\phi(x) \mid _{x = F^{-1}_\phi(s)}}{f_\phi(F^{-1}_\phi(s))}ds\,.
\end{align}
For any $u \sim \mathcal{U}(0,1)$, we have that $F^{-1}(u)  \stackrel{d}{=} X^\phi$, and the theorem follows by interpreting the integral as an expectation over a uniform random variable.
\end{proof}

The gradient formula requires knowing $F_\phi(x)$ and $\nabla_\phi F_\phi(x)$, but in general there is no analytical expression for these quantities. We instead take advantage of the mini-batch environment and use samples to generate a Gaussian kernel density estimator $\hat F_\phi(x)$ to approximate $F_\phi(x)$. Specifically, suppose we have a sample of $N$ points of $X^\phi$ given by $\{x^{(1)}_\phi, x^{(2)}_\phi, ... x^{(N)}_\phi\}$, then
$$
\hat F_\phi(x) = \frac{1}{N} \sum^N_{i=1} \Phi\left(\frac{x-x_\phi^{(i)}}{h}\right) 
$$
where $\Phi(\cdot)$ denotes the cdf of a  standardised Gaussian and $h$ is a bandwidth selected, see, e.g., \cite{gramacki2018nonparametric} for discussions around selecting bandwidths. In our implementations we use Silverman's rule divided by two : $h= 0.53\,\widehat\sigma\, N^{-1/5}$ with $\widehat\sigma$ being the sample standard deviation. Using the chain rule, we then have
$$
\nabla_\phi F_\phi(x) \approx \nabla_\phi \hat F_\phi(x) = - \frac{1}{N} \sum_{i=1}^N \tfrac1h\Phi'\left(\frac{x-x_\phi^{(i)}}{h}\right)\nabla_\phi x_\phi^{(i)}\,.
$$
Substituting this expression into Theorem \ref{thm: Non-Rob Grad}, and estimating expectation with the sample mean from the same mini-batch of data, the gradient can be approximated by
$$ \nabla R^U_g [X^\phi] \approx - \frac{1}{N} \sum^{N}_{i=1} \left[ U'(x^{(i)}_\phi) \;
\gamma\left(\hat F_\phi(x^{(i)}_\phi) \right)
\;\frac{\sum_{j=1}^N 
\Phi'\left((x^{(i)}_\phi-x_\phi^{(j)})/h\right)
\nabla_\phi x_\phi^{(j)}}{\sum_{j=1}^N 
\Phi'\left((x^{(i)}_\phi-x_\phi^{(j)})/h\right)} \right].
$$
The remaining task is to obtain $\nabla_\phi x_\phi^{(j)}$. These may be computed using back-propagation numerically and there is no need for explicit formulae.

\subsection{Robust Problem}
The steps for solving the robust problem are derived in Jaimungal et al. \cite{jaimungal2022robust}. We briefly summarise the the main results here. Denote the cdf and pdf of $X^\phi$ by $G_\phi$ and $g_\phi$, respectively. Similarly, denote the cdf and pdf of $X^\theta$ by $F_\theta$ and $f_\theta$, respectively. We define $c[X^\theta, X^\phi] = ((d_p[X^\theta, X^\phi] - \epsilon^p)_+$ to denote a penalty for being outside of the desired Wasserstein ball. To enforce the constraints, the  augmented Lagrangian approach (which utilizes both a Lagrange multiplier and a quadratic penalty, see, e.g., \cite{birgin2014practical}, Chapter 4) may be used as the loss function, and we optimise
$$
L[\theta, \phi] = \R^U_g[X^\theta] + \lambda \,c[X^\theta, X^\phi] + \frac{\mu}{2}\left(c[X^\theta, X^\phi] \right)^2.
$$
Here, $\mu$ is a penalty constant that is to be updated as one traverses the loss function along with the Lagrange multiplier $\lambda$.

\begin{theorem}[Inner Gradient Formula \cite{jaimungal2022robust}]
\label{thm: inner}
Under the assumptions of  Theorem \ref{thm: Alternate RDEU}, let $X^\phi_c$ be a reordering of the realisations of $X^\phi$ so that $(X^\theta, X^\phi)$ are comonotonic. Then the inner gradient is
\begin{equation}
    \nabla_\theta L[\theta, \phi]= \E \left[ \left\{ U'(X^\theta) \gamma\left(F_\theta(X^\theta) \right) - p \Lambda \left| X^\theta - X^\theta_c \right|^{p-1} \text{sgn} (X^\theta - X^\theta_c)  \right\} \frac{\nabla_\theta F_\theta(x) \mid _{x = X^\theta}}{f_\theta(X^\theta)}\right]
\end{equation}

where $\Lambda = (\lambda + \mu c[X^\theta, X^\phi])  \mathbb{I}(d_p[X^\theta, X^\phi]>\epsilon)$ is a constant.
\end{theorem}

\begin{theorem}[Outer Gradient Formula \cite{jaimungal2022robust}]
Under the same assumptions as Theorem \ref{thm: inner}, the outer gradient is 
\begin{multline}
    \nabla_\phi L[\theta, \phi]= \E \left[ U'(X^\theta) \gamma\left(F_\theta(X^\theta) \right) \frac{\nabla_\phi F_\theta(x) \mid _{x = X^\theta}} {f_\theta(X^ \theta)} \right.\\ \left.- p \Lambda \left| X^\theta - X^\theta_c \right|^{p-1} \text{sgn} (X^\theta - X^\theta_c) \left( \frac{\nabla_\theta F_\theta(x) \mid _{x = X^\theta}}{f_\theta(X^\theta)} + \frac{\nabla_\phi G_\phi(x) \mid _{x = X^\phi}}{g_\phi(X^\phi)} \right) \right]
\end{multline}
\end{theorem}
The proof of these formula follow utilises Theorem \ref{thm: Non-Rob Grad} coupled with the results from \cite{jaimungal2022robust}.
Similar to the previous section, we use kernel density estimation to approximate $F_\theta(x)$ and $G_\phi(x)$ and use sample means in place of the expectation.

\section{Numerical Experiments}
\label{sec:results}

\subsection{The Hedged Options}

In this section, we investigate the various hedging strategies of barrier call options under the $\alpha-\beta$ risk measures defined in Section \ref{sec: special}. In particular, let $B$ represent the barrier, and $K$ represent the strike price. Then the knock-in and knock-out option payoffs are
\begin{align*}
V^{\text{knock-in}}_T &= \left(S_T -K \right)_+ \mathbb{I}\left(\min_{t<T} \{S_t\} < B\right), \qquad \text{and}
\\
V^{\text{knock-out}}_T &=
\left(S_T -K \right)_+ \mathbb{I}\left(\min_{t<T} \{S_t\} > B\right).
\end{align*}
In our experiments, the option expires at $T = 1$, $B = 8.5$, and $K = 10$. 

\subsection{The Market Model}

Asset prices $S$ are modelled using the Heston model 
\begin{equation}
    \label{eqn:Heston 1}
    dS_t = \mu \, S_t\, dt + \sqrt{\nu_t}\,S_t\, dW_t^S\,,
\end{equation}
where $\nu$ is the instantaneous variance satisfying the SDE
\begin{equation}
\label{eqn: Heston 2}
    d\nu_t = \kappa(\theta - \nu_t)\, dt + \xi \,\sqrt{\nu_t} \,dW_t^\nu,
\end{equation}
and where $(W_t^S,W_t^\nu)_{t\ge0}$ are correlated Brownian motions under the physical measure $\mathbb{P}$.
The model parameters carry the following interpretations: $S_0$ (initial asset price), $v_0$ (initial variance), $\mu$ (drift), $\kappa$ (the rate at which $v_t$ reverts to the long run average), $\theta$ (long run average variance), $\eta$ (the volatility of volatility), and $\rho$ (the correlation between $W^S$ and $W^\nu$).

We stimulate the Heston model using $200$ time steps, however, the agent is only allowed to trade once every 4 steps, i.e. at $(t_i)_{i=0,\dots, 49}$ where $t_i := \frac{4\,i}{200}$. The model parameters in the simulation are as follows
\[
S_0 = 10, \quad v_0 =0.3^2, \quad \mu = 0.08, \quad \kappa = 3, \quad \theta = 0.3^2,
\quad \eta = 2, \quad \text{and} \;\;\rho=-0.5.
\]
When comparing the robust and non-robust models, we simulate prices from the Heston model, but where some parameters are modified.

\subsection{The Black-Scholes Model}

To evaluate the performance and interpret our neural network results, we compare to classical results derived from the Black-Scholes model. Recall that the model assumes asset prices are geometric Brownian motions. In particular, $S$ satisfies the SDE
\[
dS_t = \mu S_t dt + \sigma S_t dW_t.
\]
Strictly speaking, the Black-Scholes hedge is not for the purpose of minimising a risk measure, nor does its derivation include transaction costs, however, the resulting hedge serves as a useful benchmark for evaluating our strategies. 

It is well known that, when interest rates are zero, the price of a European call option at time $t$ in the Black-Scholes model is given by
\[
C_{BS} (S, K) = S \,\Phi(d_+) - K\,\Phi(d_-),
\]
where $d_\pm= \big(\log \frac{S}{K} \pm \frac{\sigma^2}{2} (T-t)\big)/\sigma \sqrt{T-t}$. 
As well, the price of knock-in and knock-out options, before asset prices touch the barrier are provided, respectively, by \cite{Westermark}
$$ C_{\text{Knock in}} = \frac{S}{B} C_{BS} \left( \frac{B^2}{S}, K \right) \qquad \text{and} \qquad 
C_{\text{Knock out}} = C_{BS} - C_{\text{Knock in}}$$
For knock-in calls, once the barrier is breached, the option becomes a regular call and has value $C_{BS} (S, K)$. For knock-out calls, once the barrier is breached, the option is worthless.
For knock-in calls, the Black-Scholes delta hedge is given by
\begin{align*}
    \Delta_{\text{Knock-in}}(S) = \frac{1}{B} C_{Call} \left( \frac{B^2}{S}, K \right) - \frac{B}{S} \Delta_{Call} \left( \frac{B^2}{S}, K \right) \,.
\end{align*} 
The Black-Scholes delta hedge for a knock-out call is $\Delta_{\text{Knock-out}}(S) = \Phi(d_+)-\Delta_{\text{Knock-in}}(S)$.
In our experiments, when we hedge using the Black-Scholes hedge, we do so at the same set of discrete hedging times that the RL agent uses. Moreover, we use the  $\sigma$ estimated by stimulating price paths $S_t$ under Heston to generate a distribution for $\log(S_T)$ and computing its standard-deviation divided by $\sqrt{T}$. Thus, obtaining a Black-Scholes model that matches the variance of the Heston model.

\subsection{Hedging Features}
We now consider the agent's problem as posed in Section \ref{sec: problem}. The strategy,  $\phi = (\Delta_i)_{i = 0, 1 ... N-1}$, is parametrised by a fully connected feed forward neural network with 5 hidden layers each with 35 neurons and SiLU activation functions\footnote{The SiLU activation function is given by $f(x)=x/(1+e^{x})$}. The output layer has a $\tanh$ activation function that clips strategies to the range $\Delta_i\in[-2,2]$ to ensure the agent does not over-leverage. 
For test cases with no transaction costs, the neural net uses the following five features
$$\left\{t,\, S_t, \,M_t, \,\mathbb{I}\!\left[M_t > B\right]\right\}.
$$ 
where $M_t:=\min_{u\le t} \{S_u\}$.
When we incorporate transaction costs, we add one more feature: the previous holdings $\Delta_{t-1}$.

While $\mathbb{I}[M_t < B]$ can be determined through $\min_{u<t} \{S_u\}$, each serve a different purpose in the neural network. In most cases, the optimal strategies have a jump discontinuity once the barrier is breached, but, as activation functions are smooth, neural networks have difficulties modelling such discontinuities. Thus, we include the indicator of whether the barrier was breached or not. We include the minimum asset price up to that point in time, as under the $\alpha-\beta$ risk measures its effect on optimal strategies is a priori unclear. Under the Black-Scholes model, e.g., only time, price, and the indicator matters.  

For the robust problem, a second neural net distorts $X^\phi$ and it is parametrised by a fully connected feed forward neural network with 1 hidden layer of 10 neurons and SiLU activation functions. It takes in a realisations of $X^\phi$ and outputs realisations of $X^\theta$. For all experiments, we use the Wasserstein-$1$ distance and set the radius of the ball to $0.02$.

\subsection{CVaR}
\label{sec: Cvar}
In this section, let $\alpha = 0.2$, $p = 1$. For these parameters, the $\alpha-\beta$ risk measure reduces to $\text{CVaR}_{0.2}$. 

\subsubsection{No Transaction Costs}
We first set transaction costs to zero to illustrate a few interesting effects that are present even in the case without transaction costs.
\begin{figure}[H]
  \centering
  \includegraphics[width=\textwidth]{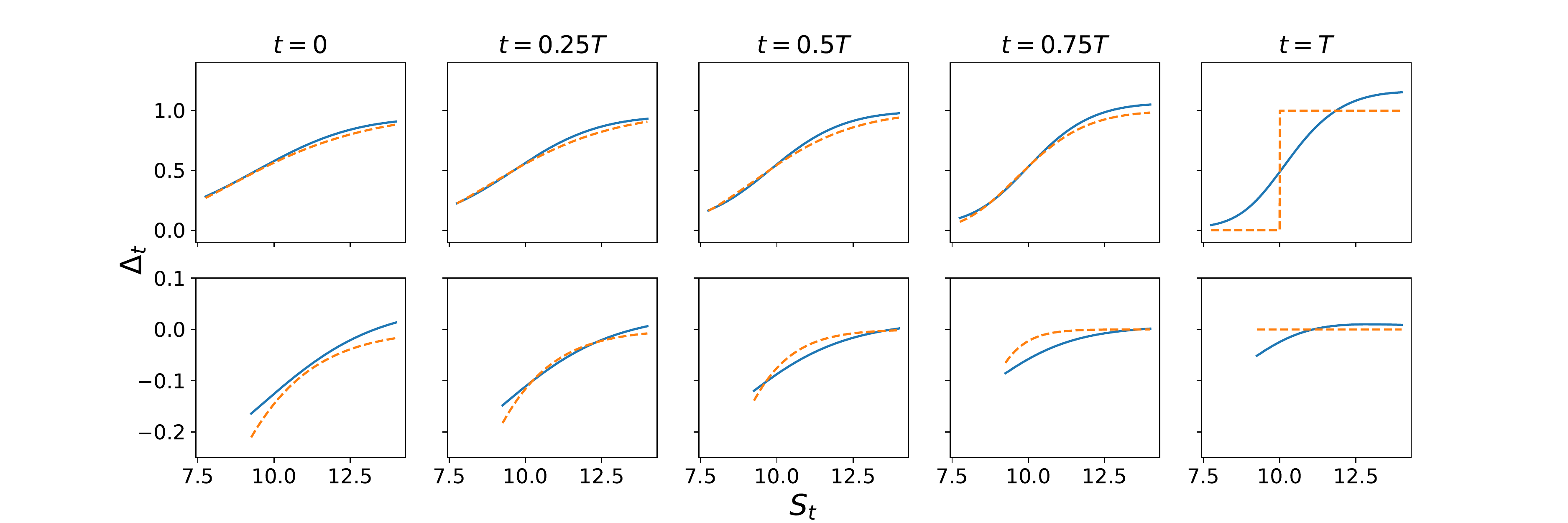}
  \caption{A comparison between trained non-robust strategies (blue) and Black-Scholes delta hedging (orange) for the knock-in call option. The top panel shows post-knock in strategies, while the bottom panels shows pre-knock in strategies. }
  \label{fig:In CVaR non robust}
\end{figure}
\begin{figure}[H]
  \centering
  \includegraphics[width=\textwidth]{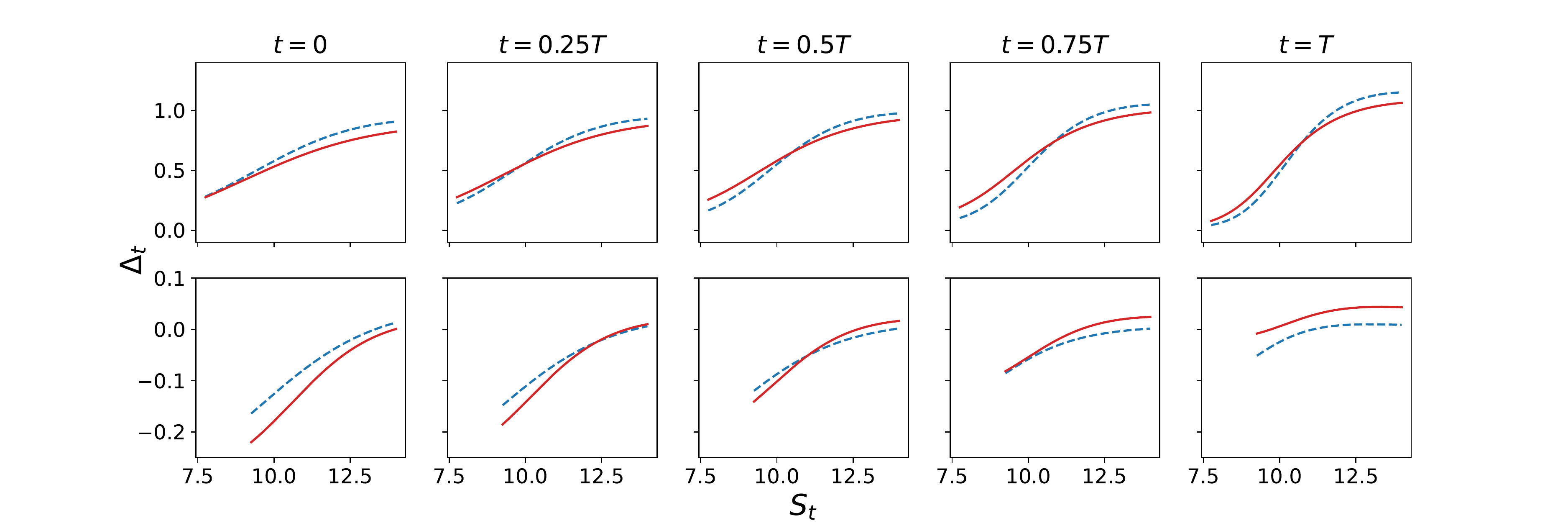}
  \caption{A comparison between robust (red) and non-robust (blue) strategies for the knock-in call option. The top panel shows post-knock in strategies, while the bottom panels shows pre-knock in strategies. }
  \label{fig:In CVaR robust}
\end{figure}

Figures \ref{fig:In CVaR non robust} and \ref{fig:In CVaR robust} shows the optimal strategy learned by the agent. The top panels show the optimal hedge post knock-out (i.e., once the barrier is breached), while the bottom panels show the pre-knock-out (i.e., before the barrier is breached) hedges. While $M_t$ is a feature, the strategies for CVaR only depend on $\mathbb{I}[M_t < B]$ but do not depend on the exact value of $M_t$. In the lower panels, the hedges stop at $M_t= 8.5=B$ corresponding to the option barrier. Figure \ref{fig:In CVaR non robust} shows that, for CVaR, the non-robust strategy are qualitatively and quantitatively fairly similar to the Black-Scholes' delta-hedging. Moreover, Figure \ref{fig:In CVaR robust} shows that, for CVaR, the robust and non-robust are similar, but there are some notable differences. The robust agent does hedges more when the option is out-of-the-money and less when the option is in-the-money. By replicating the option payoff less closely, the agent is better able to handle market uncertainties. 

Figure \ref{fig:Out CVaR} shows the analogous strategies for knock-out call options for the robust and non-robust cases. They are very similar to Black-Scholes hedges (not shown). Naturally, the agent only hedges when the option has not yet been knocked-out, otherwise, once it is knocked-out, the agent holds no position in the robust case, but holds a small positive position in non-robust case. This is because the agent is hoping to gain some positive return from the asset.

\begin{figure}[H]
  \centering
  \includegraphics[width=\textwidth]{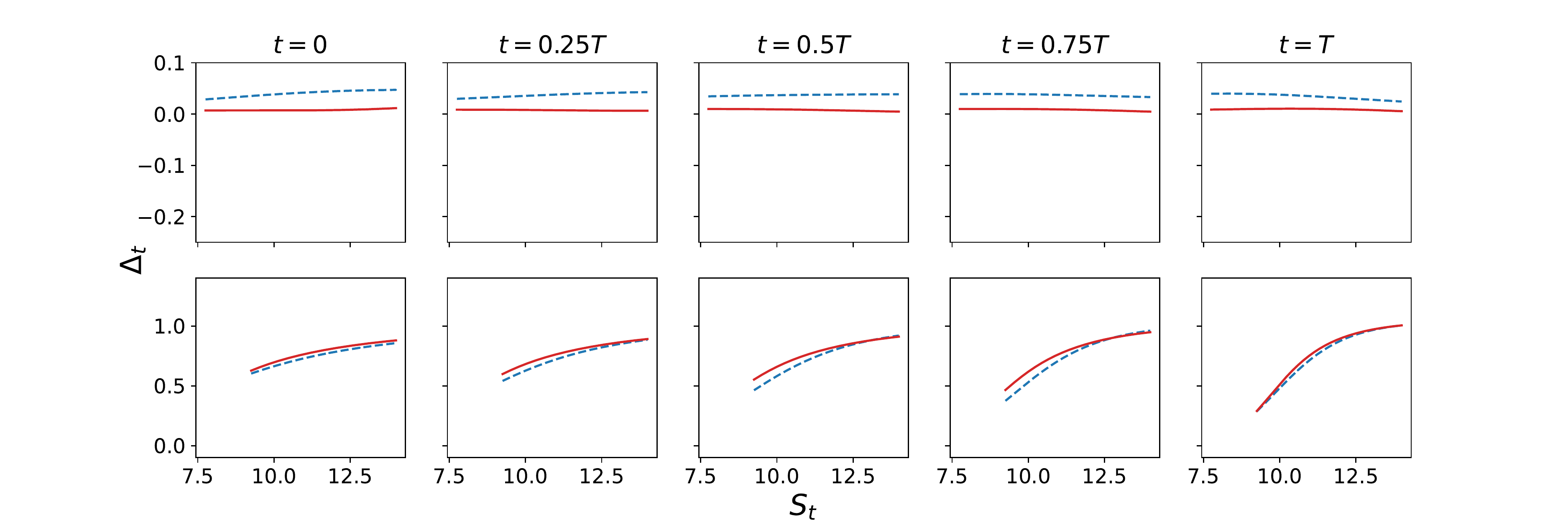}
  \caption{A comparison between robust (red) and non-robust (blue) strategies for the knock-out option. The top panel shows post-knock out strategies, while the bottom panels shows pre-knock out strategies. }
  \label{fig:Out CVaR}
\end{figure}

\subsubsection{The Effect of Transaction Costs}
\label{section: CVaR Rob}

Next, we introduce transaction costs and set $c = 0.01$ in Equation \ref{eqn:wealth}. In this case, it is challenging to visualise the strategy as a function of features, as the strategy now becomes a function of $t$, $S_t$, $M_t$, and $\Delta_{t-1}$.
Instead, we depict the optimal strategy for a few sample paths. To this end, Figure \ref{fig:CVar Trans paths} shows the optimal hedge along two sample paths of the asset process -- both with and without the transaction costs. The two strategies co-move. The right most panel shows a histogram of total variation of hedge position across 5000 paths
As the figure shows, the case with transaction costs has less total variation, indicating that the agent trades less in this case. Both total variations have spikes near zero because if the barrier is never breached the agent trades very few shares.
\begin{figure}[H]
     \centering
     \begin{subfigure}{0.6\textwidth}
         \centering
         \includegraphics[width=\textwidth]{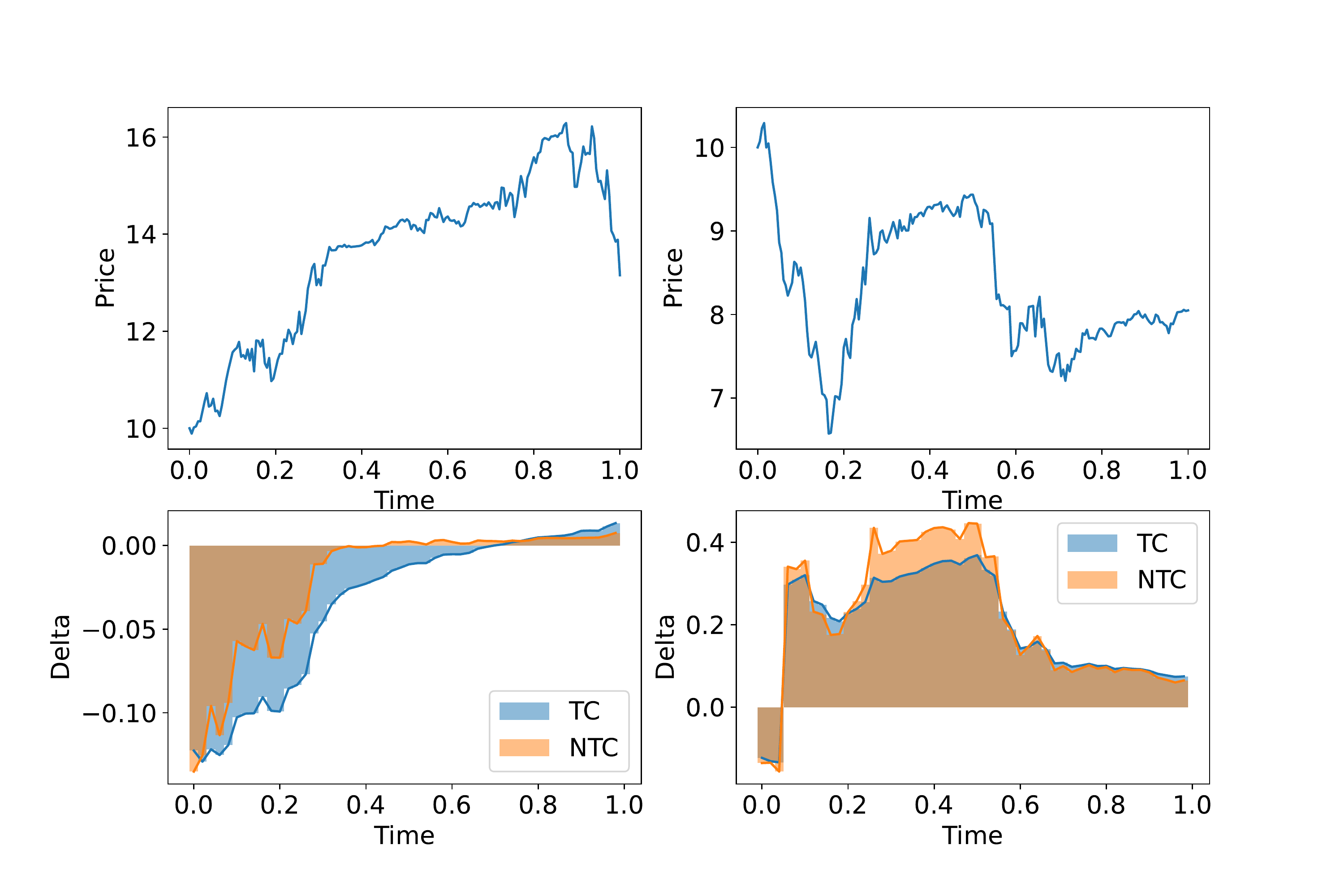}
         \caption{Optimal $\Delta_t$ for two sample asset paths.}
     \end{subfigure}
     \hspace{0.5cm}
     \begin{subfigure}{0.3\textwidth}
         \centering
         \includegraphics[width=\textwidth]{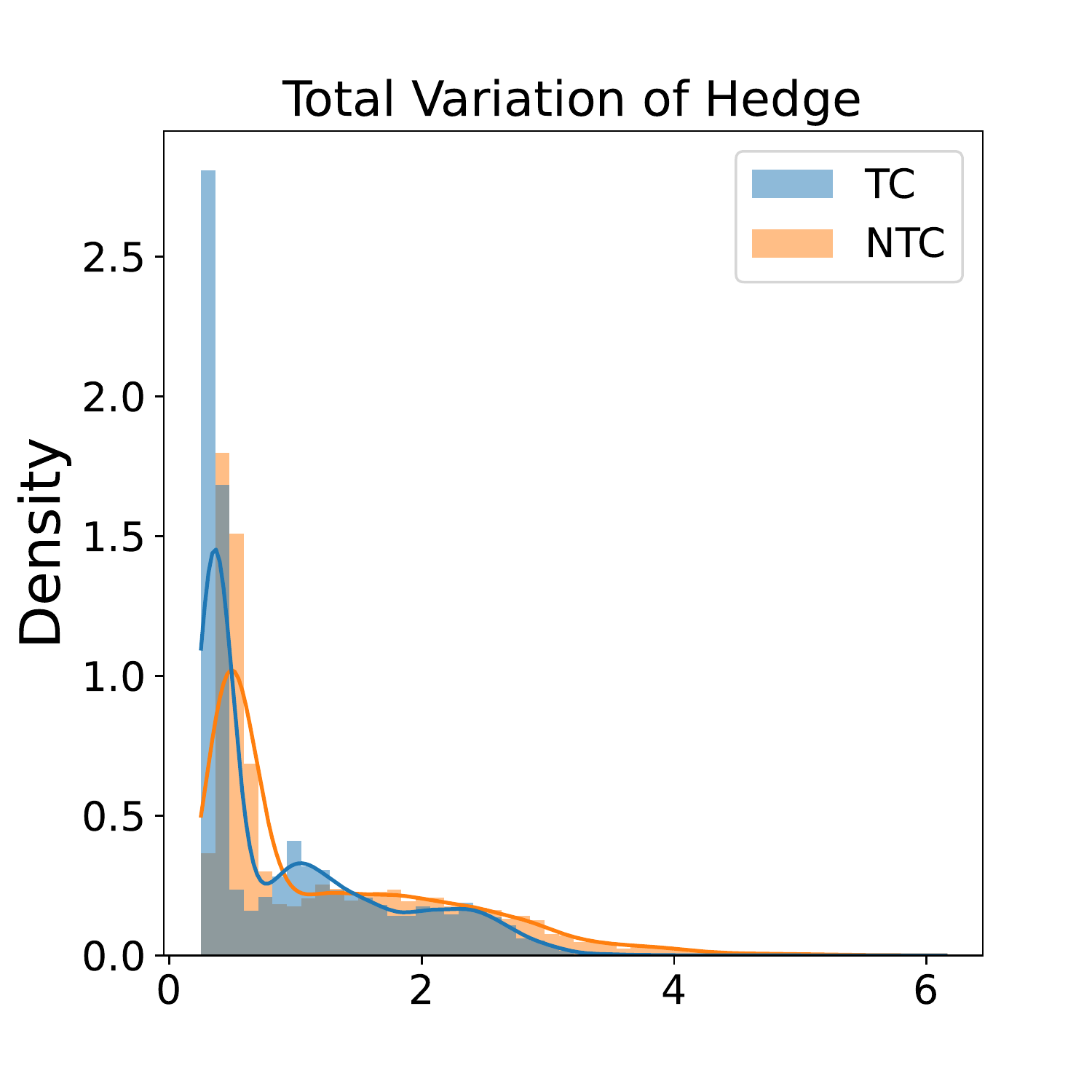}
         \caption{Histogram of total variation of $(\Delta_t)_{t = 0, 1 \dots N-1}$
         across multiple asset paths.}
     \end{subfigure}
        \caption{Strategies for hedging the down-and-in call option with transaction costs (TC) and with no transaction costs (NTC) built into the model.}
        \label{fig:CVar Trans paths}
\end{figure}

Figure \ref{fig:CvaR compare} shows a comparison of P\&L between the robust strategy and the Black-Scholes option hedging strategy, under uncertainty. All training price paths are generated using the parameters $S_0 = 10, v_0 =0.3^2, \mu = 0.08, \kappa = 3, \theta = 0.3^2, \eta = 2$, and $\rho=-0.5$. For the Black-Scholes models, we ignore the transaction cost and estimate $\sigma$ using the method described above. Furthermore, we assume that the model was misspecified with $\kappa = 1$ and $\rho = -0.1$ in reality. When there is model uncertainty, the robust strategy outperforms in both cases.
\begin{figure}[H]
     \centering
     \begin{subfigure}{0.42\textwidth}
         \centering
         \includegraphics[width=\textwidth]{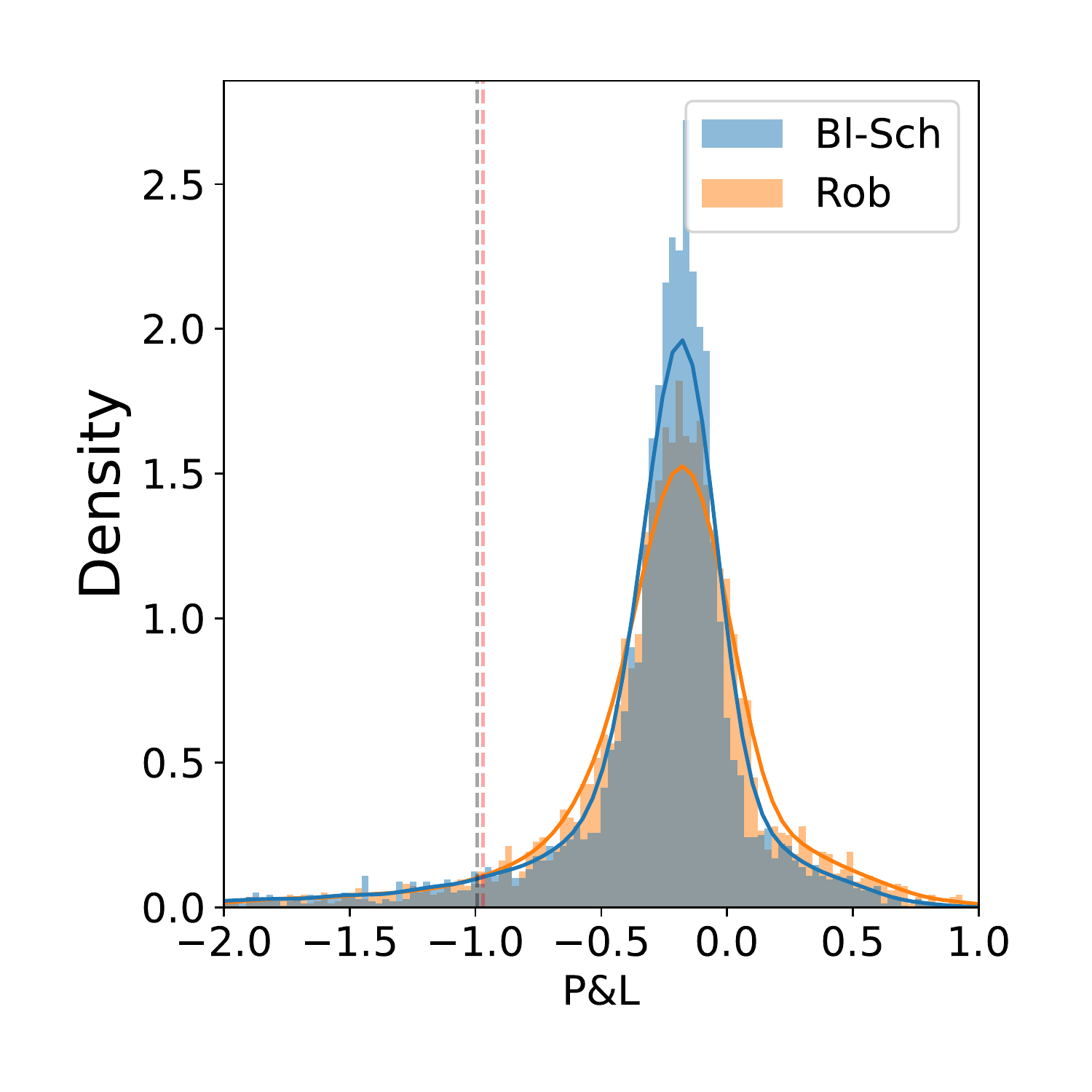}
         \caption{Knock-in Option}
     \end{subfigure}
     \begin{subfigure}{0.42\textwidth}
         \centering
         \includegraphics[width=\textwidth]{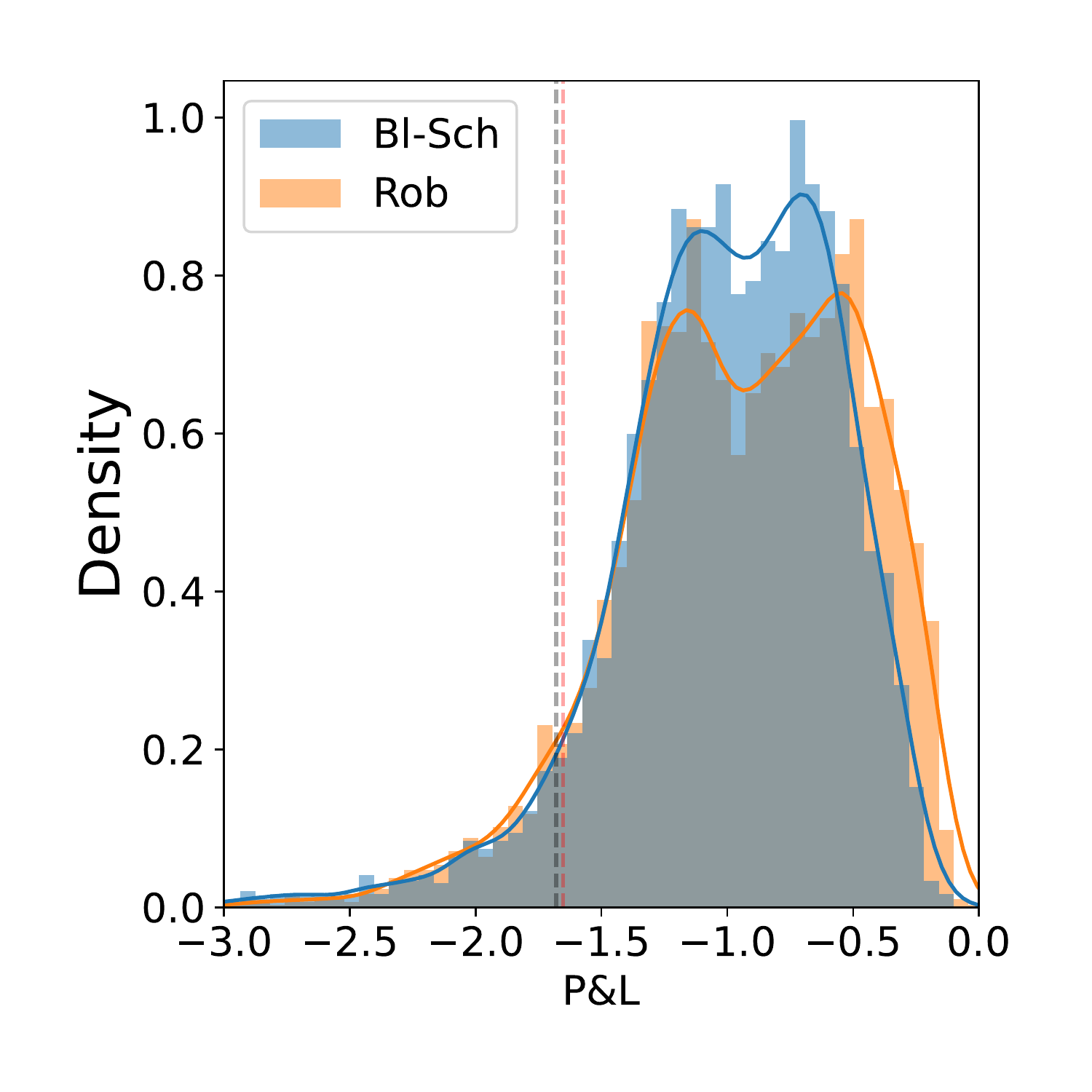}
         \caption{Knock-out option}
     \end{subfigure}
        \caption{Comparison of Black-Scholes option hedging (Bl - Sch) and robust (Rob) P\&L with transaction cost when model is misspecified with $\kappa = 1$ and $\rho = -0.1$ in reality. CVaR$_{0.2}$ is also shown.}
        \label{fig:CvaR compare}
\end{figure}

\subsubsection{Pricing}
The methodology we develop may also be used to provide a robust price for options. Suppose the agent wishes to target a $\text{CVaR}_{0.2}$ of $-0.5$. As CVaR is translation invariant, after learning a strategy, the agent only needs to adjust the initial option price to reach the required target -- the optimal strategy will not be altered -- the result price is denoted $P^{Rob}$.

Next, we compare the price induced by the robust strategy to the price induced by the Black-Scholes strategy. We cannot, however, simply use the Black-Scholes option price as it does not target a CVaR, nor does it include transaction costs. Instead, we generate the P\&L of the Black-Scholes and adjust the price until the P\&L$\text{CVaR}_{0.2} = -0.5$ --- we denote this price by $P^{BS}$. As a benchmark comparison, we use the classical Black-Scholes barrier option price -- denoted $C^{BS}$.

Table \ref{table: pricing} shows the effect of robust pricing. We train our model with the same parameters as Section \ref{sec: Cvar} and assume that in reality $\kappa = 1$ and $\rho = -0.1$ (instead of $\kappa = 3$ and $\rho = -0.5$). The agent, however, has no way of knowing the true $\sigma$, and must price using the misspecified model. From our results, targeting $\text{CVaR}_{0.2} = -0.5$ is a reasonable target, and the generated prices are close to $C_{BS}$. As the table shows, the classical Black-Scholes price is the lowest,  the price induced by the Black-Scholes hedge (but, accounting for transaction costs, and s.t. $CVaR=-0.5$) is higher, and finally, as expected, the robust price is the highest.
\begin{table}[H]
\centering
\begin{tabular}{@{}p{0.12\textwidth}*{4}{L{\dimexpr0.22\textwidth-2\tabcolsep\relax}}@{}}
\toprule\toprule
& \multicolumn{2}{c}{Knock-in Call} &
\multicolumn{2}{c}{Knock-Out Call} \\
\cmidrule(r{4pt}){2-3} \cmidrule(l){4-5}
&  Price & CVaR$_{0.2}$ & Price & CVaR$_{0.2}$\\
\midrule
$P^\text{Rob}$  & 0.32258 & $-0.68025$ & 1.16839 & $-0.51235$ \\
 $P^{BS}$  & 0.31434  & $-0.74954$ & 1.11398 & $-0.60988$ \\
 $C^{BS}$  & 0.27300  & - & 0.98123 & - \\
\bottomrule\bottomrule
\end{tabular}
\\[0.5em]
\caption{Comparison of different pricing schemes. The left column under each header shows the price outputted by the misspecified model targeting CVaR$_{0.2} = -0.5$. The right column shows the CVaR of a strategy charging their respective price under the actual model. }
\label{table: pricing}
\end{table}

\subsection{\texorpdfstring{$\alpha$-$\beta$}{a-b} risk measures}
\label{sec: ab}

We next consider $\alpha-\beta$ risk measures with $\alpha = 0.1$, $\beta =0.9$, and $p = 0.7$ -- see the middle panel of Figure \ref{fig:gamma}. This corresponds to an agent who cares about gains to losses at a ratio of $3 : 7$.

\subsubsection{No Transaction Costs}
\begin{figure}[H]
  \centering
  post knock-in \hspace{11em}
  pre knock-in
  \\
  \includegraphics[width=0.85\textwidth]{Down_in_ab_non_robust.pdf}
  \caption{Non robust hedging strategies for a down-and-in call option. The graph show total asset holdings at t = 0.5T.}
  \label{fig: ab in non robust}
\end{figure}
First we consider no transaction costs. Figure \ref{fig: ab in non robust} shows the hedging portfolios for a knock-in call in the non-robust case. These strategies are best described as a combination of asset trading and option hedging. When the option is not yet knocked-in, the agent ignores the embedded call option and invests in the asset. In this case, when the asset price is significantly above the barrier, they long the asset, but the strategy has a built in stop-loss. As the minimum asset price drops, they start unwinding their position. If prices drop to around \$9, for example, the agent refrains from taking a long position in the asset for the remainder of the session. This strategy allows the agent to take advantage of the positive drift in prices, but also places a cap on the losses incurred from trading the asset. If the barrier is breached, the agent quickly adjusts to option hedging. 

\begin{figure}[H]
  \centering
    post knock-in \hspace{11em}
  pre knock-in
  \includegraphics[width=0.85\textwidth]{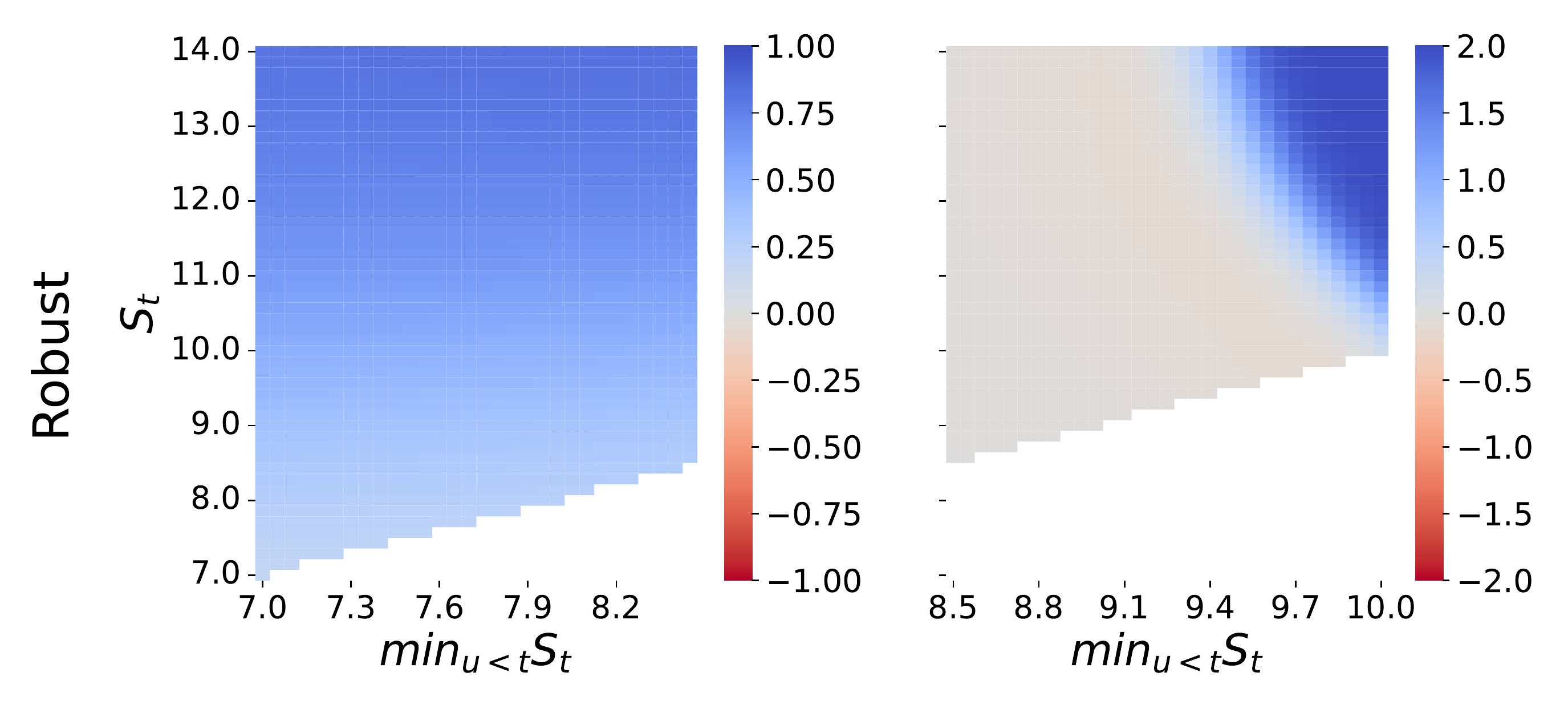}
  \caption{Robust hedging strategies for a down-and-in call option. The graph show total asset holdings at t = 0.5T}
  \label{fig: ab in robust}
\end{figure}
Figure \ref{fig: ab in robust} shows the corresponding robust strategy. In this case, the agent becomes more conservative in seeking gains. They begin by buying a small number of shares. If the stock price rises, the agent has made a profit and seeks additional gains by purchasing more shares. Compared with the non-robust case, the agent sets a higher stop loss and generally exits from the long position earlier. Similar to the CVaR cases, when the option is knocked-in, the robust strategy holds less in the stock than the non-robust strategy in an attempt to mitigate model uncertainty.

\begin{figure}[H]
  \centering
  post knock-out \hspace{9em}
  pre knock-out
  \\
  \includegraphics[width=0.85\textwidth]{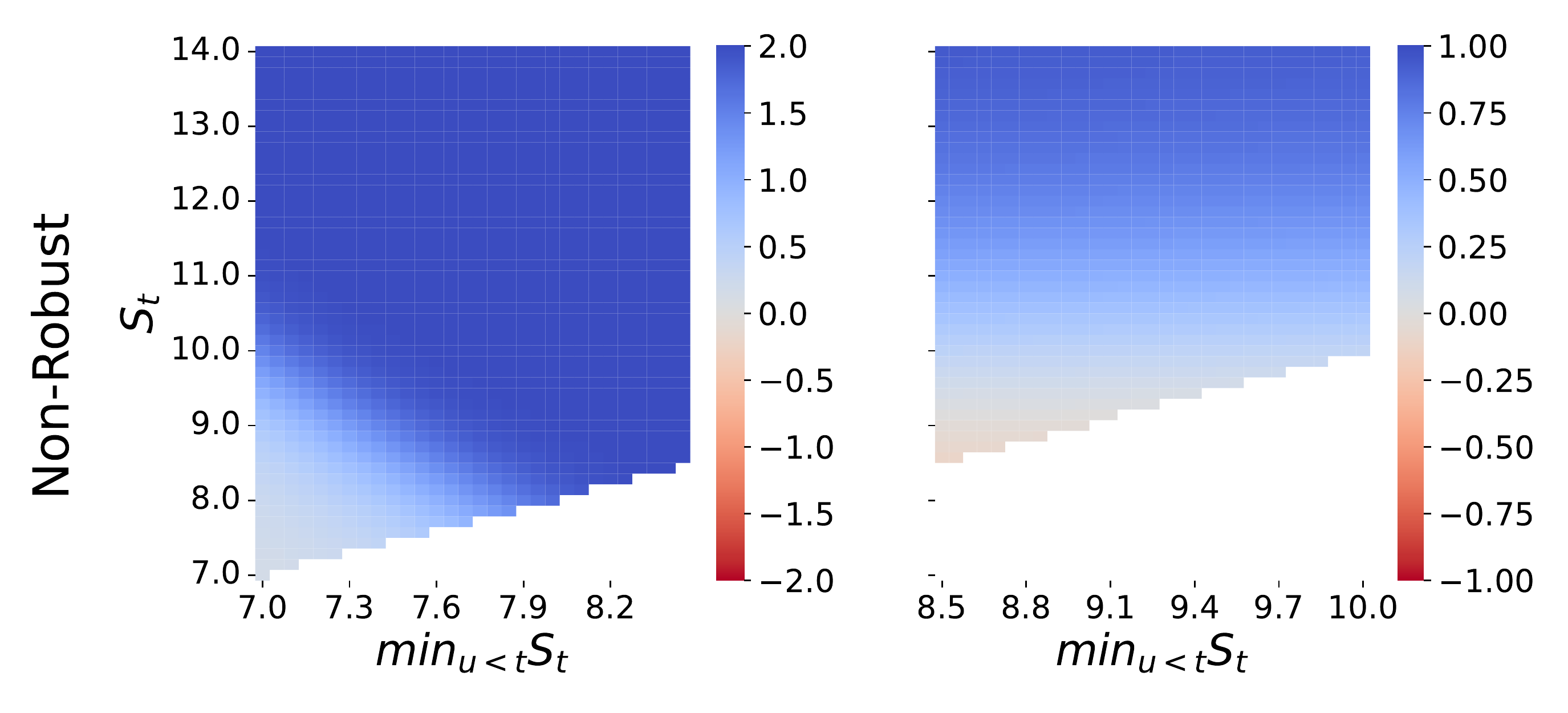}
  \caption{Non robust hedging strategies for a down-and-out call option. The graph show total asset holdings at t = 0.5T}
  \label{fig: ab out non robust}
\end{figure}
Figure \ref{fig: ab out non robust} shows analogous results for the knock-out call. When the option is knocked-out (the left panel), the strategy is to go long the asset, but as in the previous examples, it is paired with a stop-loss strategy. When we robustify the strategy, as shown in \ref{fig: ab out robust}, the risk in asset trading when knocked-out is greater than the agent can bear, and they simply hold nothing until the end of the trading horizon. In the region before being knocked-out, however, the agent focuses on option hedging. This is dramatic change is due to a phase transition that occurs in investor's behaviour as they change the balance between what is more important gains or losses. We discuss this phenomenon in more detail in Section \ref{sec: phase transition}.
\begin{figure}[H]
  \centering
  post knock-out \hspace{9em}
  pre knock-out
  \\
  \includegraphics[width=0.85\textwidth]{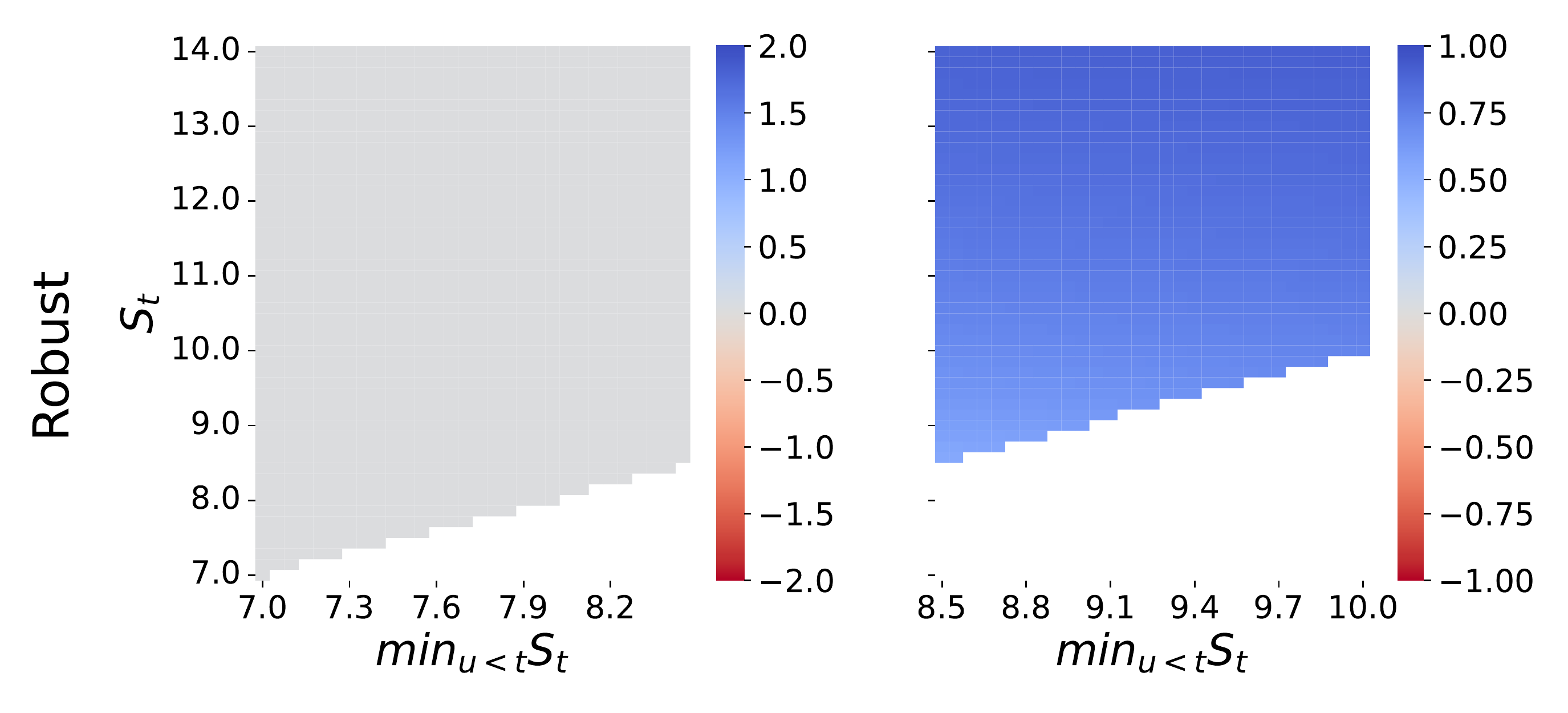}
  \caption{Robust hedging strategies for a down-and-out call option. The graph show total asset holdings at t = 0.5T}
  \label{fig: ab out robust}
\end{figure}

\subsubsection{Incorporating Transaction Costs and Robustness}

Next, we consider transaction costs of $c = 0.01$. In Figure \ref{fig:ab variaton}, we show the total variation of the position and compare with the no transaction case. Overall, when the agent incorporates transaction costs, they favour strategies with less variance. For example, in the knock-in case, the no-transaction-cost agent purchases close to 2 shares at $t = 0$. This strategy results in very little transactions when prices rise, however, it performs poorly performing when prices fall because the agent looses from both transaction costs and asset value. Instead, the transaction-cost agent, chooses to buy less shares at t = 0, and adjust shares as needed.  Compared to CVaR risk measures, introducing transaction costs has a larger impact on the strategies. For both knock-in and knock-out options, the agent avoids paying the upper tail risk due to the transaction costs.
\begin{figure}[H]
     \centering
     \begin{subfigure}{0.45\textwidth}
         \centering
         \includegraphics[width=\textwidth]{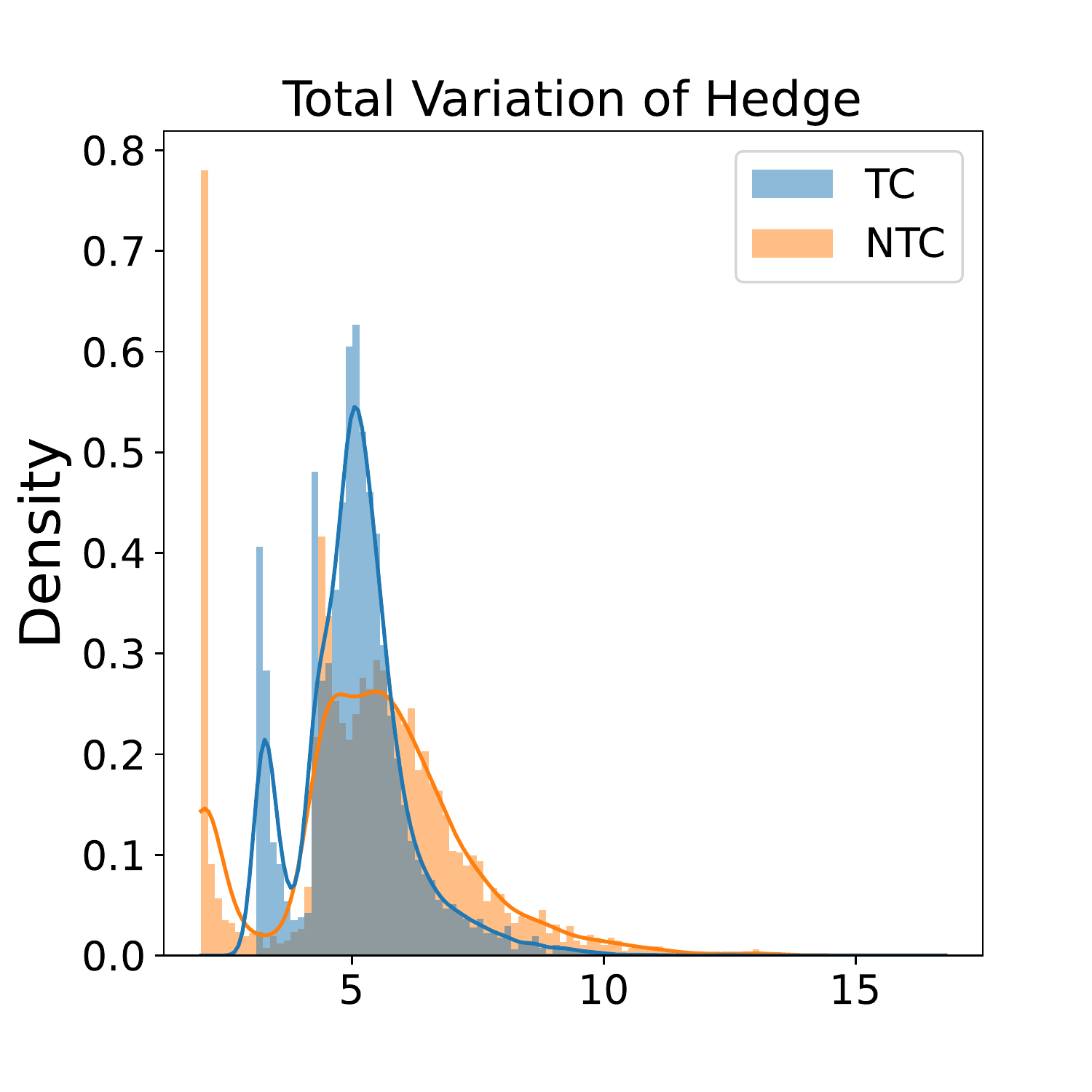}
         \caption{Knock-in option}
     \end{subfigure}
     \hspace{0.5cm}
     \begin{subfigure}{0.45\textwidth}
         \centering
         \includegraphics[width=\textwidth]{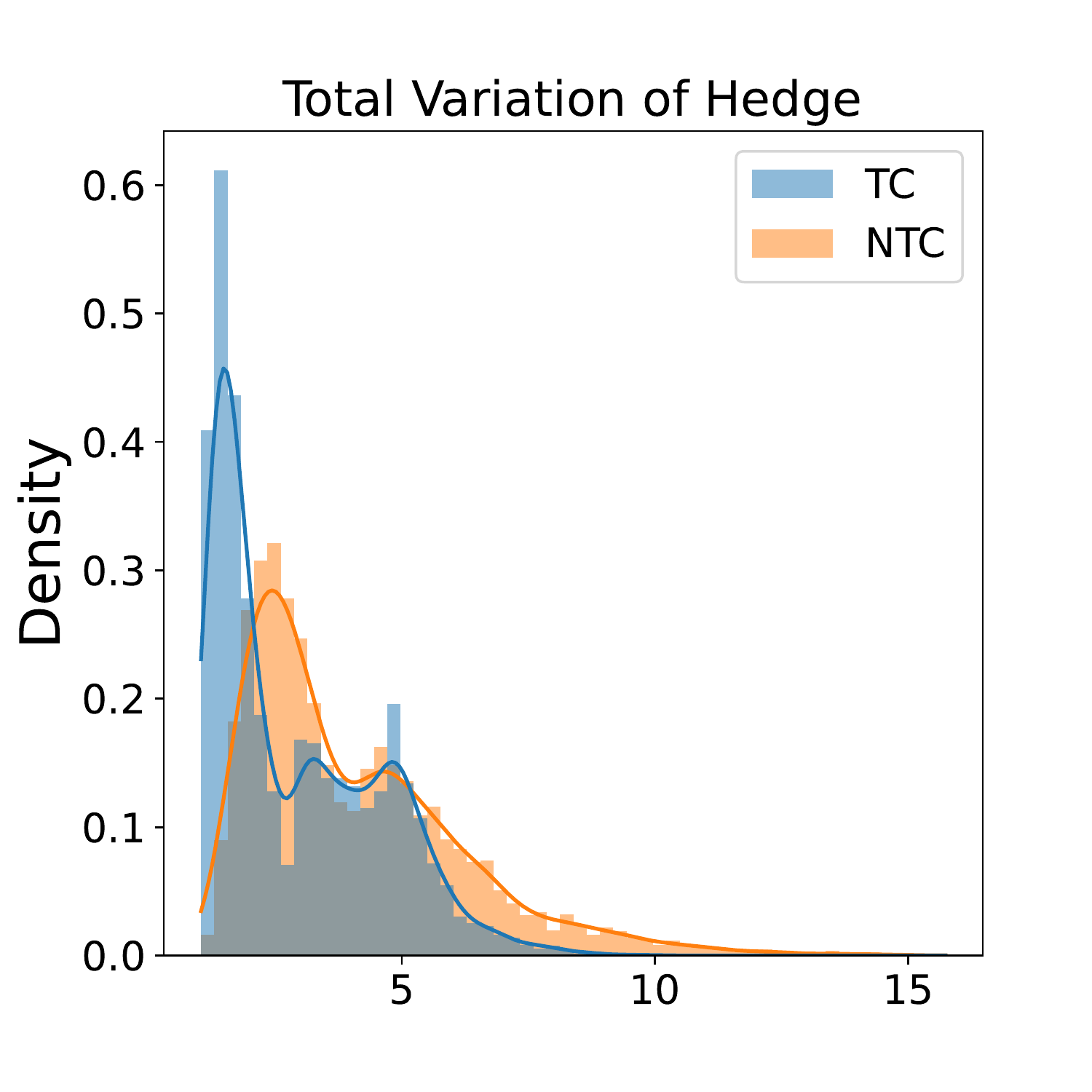}
         \caption{Knock-out option}
     \end{subfigure}
        \caption{Comparison of total Variation of the hedge position across different asset price paths in cases with transaction costs (TC) and without transaction costs (NTC).}
        \label{fig:ab variaton}
\end{figure}

Figure \ref{fig:ab compare} compares the robust and non-robust strategies. Under our framework, uncertainty need not arise from the asset price process. For example, the strategy also protects against changes in transaction costs. Models were trained using the same parameters as before. However, for testing, we increase transaction costs to $c = 0.07$. In both cases the robust strategy forgoes the upper tail gain in order to minimise losses. Since losses are more valued than gains, the agent is better off in terms of the risk measure.

\begin{figure}[H]
     \centering
     \begin{subfigure}{0.45\textwidth}
         \centering
         \includegraphics[width=\textwidth]{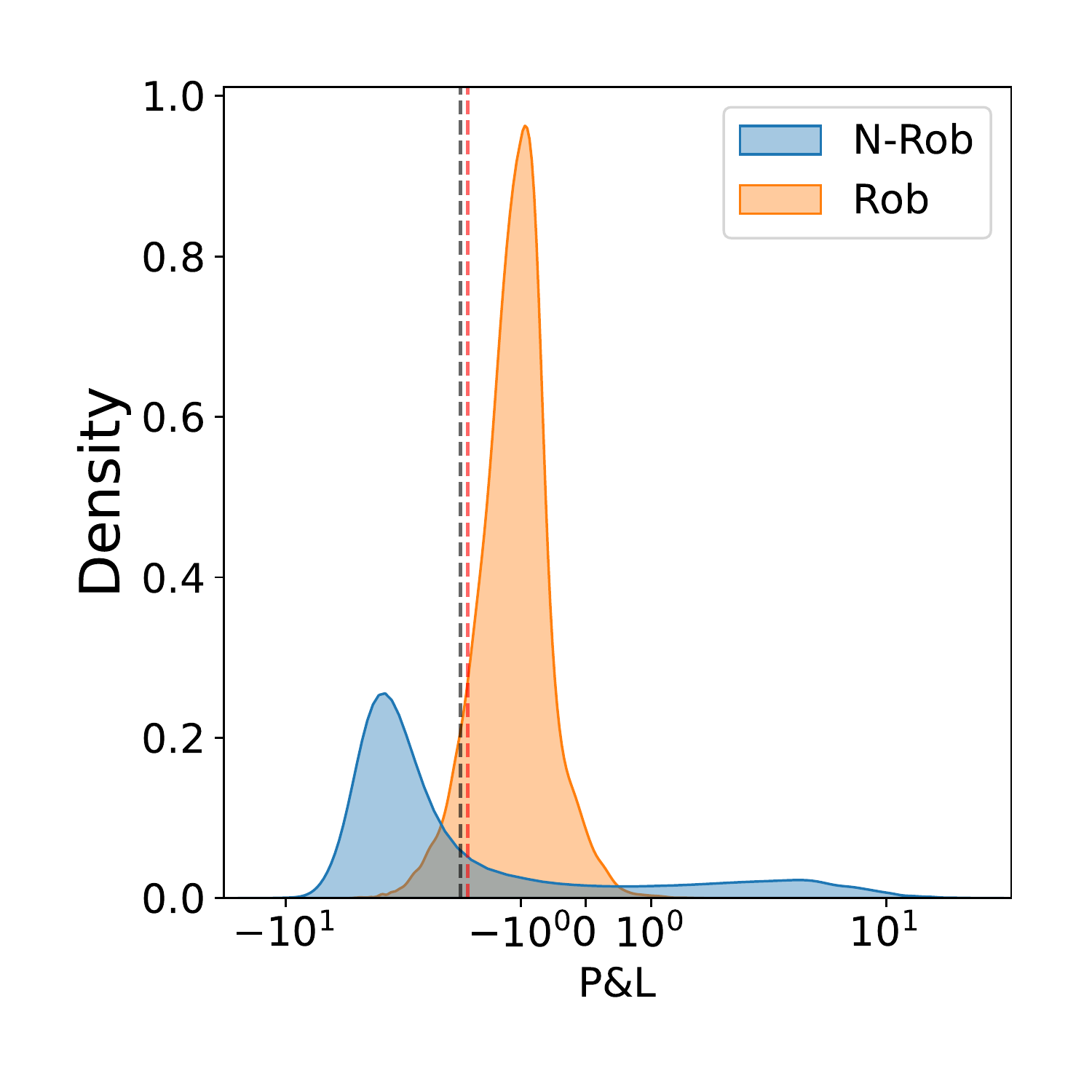}
         \caption{Knock-in Option}
     \end{subfigure}
     \hspace{0.5cm}
     \begin{subfigure}{0.45\textwidth}
         \centering
         \includegraphics[width=\textwidth]{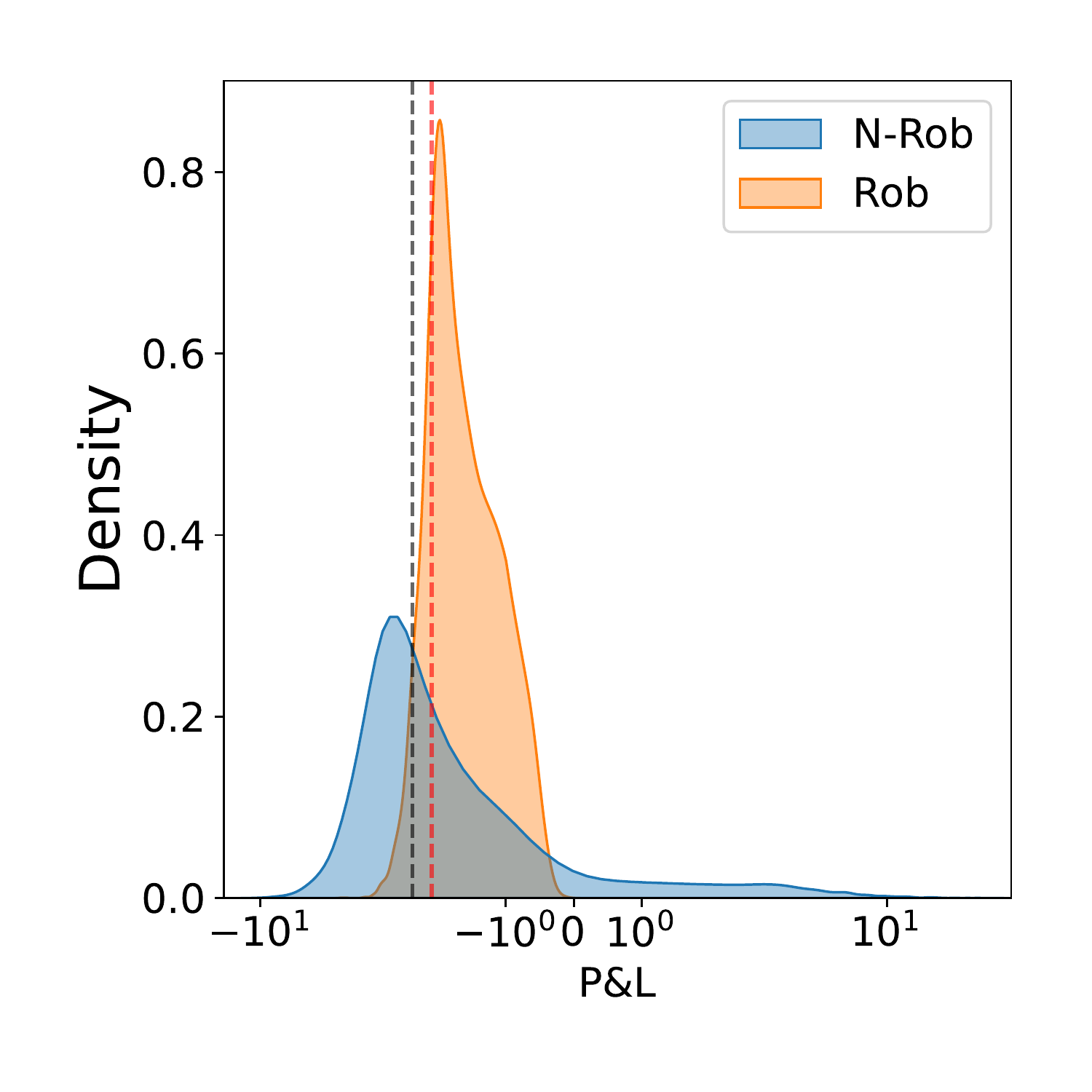}
         \caption{Knock-out option.}
     \end{subfigure}
        \caption{Comparison of non-robust(N-Rob) and robust (Rob) strategies. $-\R^U_g [X^\theta]$ is shown as dotted lines in both plots. To help with visualisations, the x-axis uses a symmetric logarithmic scale, with values inside $[-1,1]$ displayed on a linear scale}
        \label{fig:ab compare}
\end{figure}

\subsection{Phase Transition due to Agent's Preferences}
\label{sec: phase transition}
In this section, we examine how the agent's risk preferences affect the optimal strategy and study the non-robust case (as the general features we observed remain similar when we robustify). In particular, we fix $\alpha = 0.1$ and $\beta =0.9 $ and vary $p$ to examine how the strategy for the knock-in and knock-out call changes as the agent becomes increasingly risk seeking. 

Figure \ref{fig: vary p} shows the strategies for $p = 1, 0.85$ and $0.70$ along two stock paths. We assume the same training parameters as before and include transaction costs. The strategies for $p = 1$ and $0.85$ are almost identical, and the agent seeks a strategy similar to option hedging. When $p$ decreases to 0.70, however, the agents strategy differs significantly when the embedded call option is not enforced (i.e,. in the knock-out case after the barrier breached, while in the knock-in case before the barrier is breached).  For example, let us focus on the left sample path for the knock-in option. In this sample path, the asset never went below the barrier, and the $p=0.7$ agent kept a long position for the entire path, while the $p=1$ and $p=0.85$ agents held a very small short position. The more aggressive agent was seeking profits, while the more conservative was initially protecting against the potential of the option getting knocked-in, but then slowly unwound the position as the asset price grew. For the same sample path, but now looking at the knock-out option. The $p=1$ and $p=0.85$ agents took on a hedging position against the call option, and kept slowly increasing that hedge through time as the asset price grew. While the $p=0.7$ agent initially ignored hedging the call option, but as time progressed they transformed their position into a hedge position for the call which appear more and more likely to not get knocked-out.
\begin{figure}[h]
  \centering
  \includegraphics[width=0.95\textwidth]{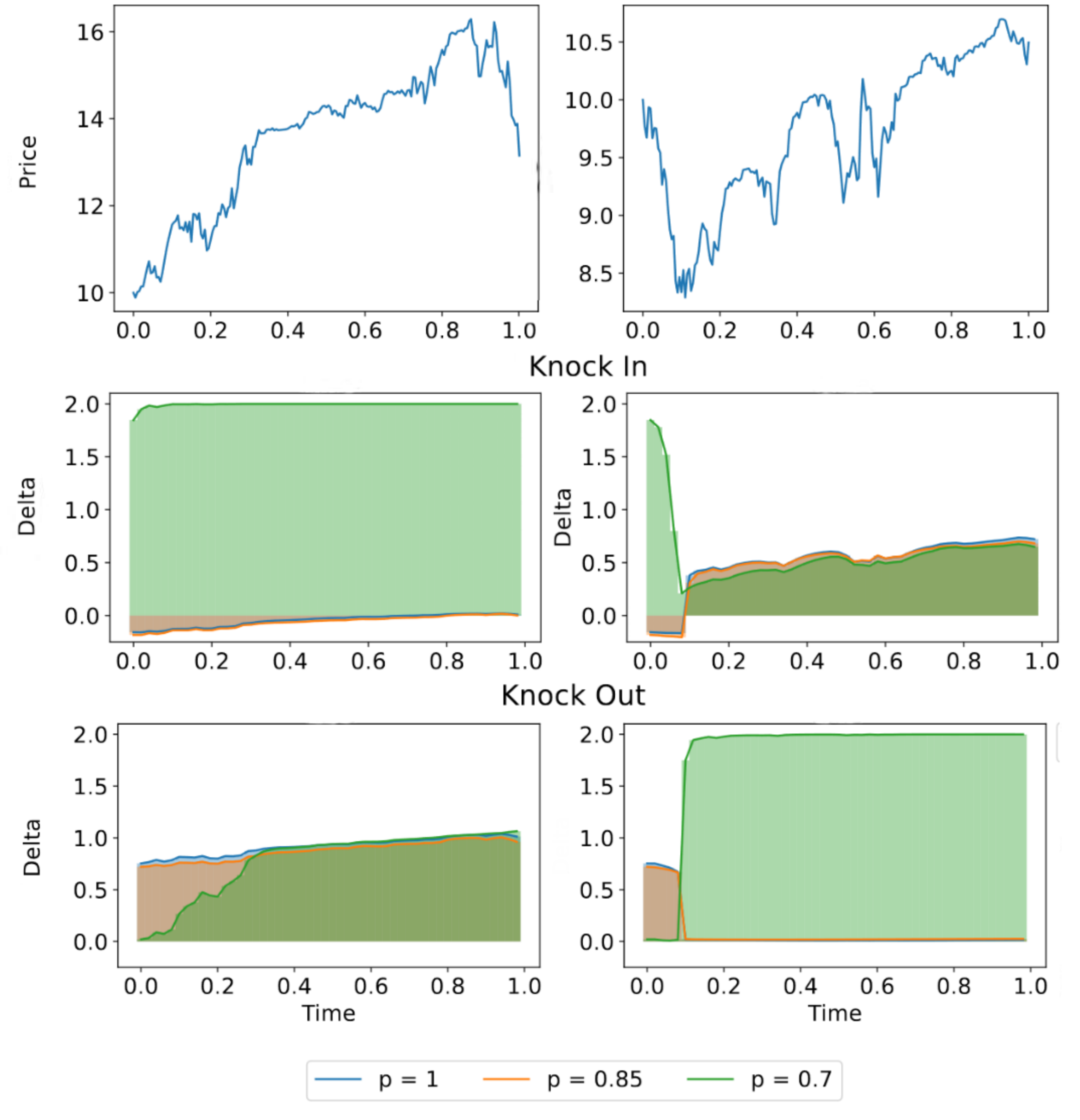}
  \caption{$\Delta_i$ along two sample paths for knock-in (above) and knock-out (below) options. Strategies shown correspond to p = 1, 0.85, and 0.70}
  \label{fig: vary p}
\end{figure}

For $\alpha-\beta$ risk measures, the transition between the agent performing option hedging (Section \ref{sec: Cvar}) to a mixed strategy of trading and hedging (Section \ref{sec: ab}) is not smooth. There seems to be a \textit{phase transition}: after a critical point $p^*$, the agent changes strategies. We conjecture similar phase transitions persists with risk-reward criteria and is not specific to the $\alpha-\beta$ risk measures.

To investigate this observation, we use the non-robust model with no transaction costs as a baseline. Recall that the $\alpha-\beta$ risk measure can be written as $\R^{\alpha \beta}_p[X^\phi] = - \frac{1}{\eta}\left(p\, \E[X^\phi \mathbb I(F(X^\phi) < \alpha)] - (1-p)\, \E[X^\phi \mathbb I(F(X^\phi) > \beta)]\right)$.
We refer to terms $\E[X^\phi \mathbb I(F(X^\phi) < \alpha)]$ and $\E[X^\phi \mathbb I(F(X^\phi) > \beta)]$ as lower and upper tail expectations (LTEs and UTEs), respectively. Figure \ref{fig:phase transition} shows the effect $p$ has on the optimal strategy. The phase transition appears to occur near $p^* \approx 0.84$ and $p^* \approx 0.81$ for knock-in and knock-out options, respectively. For $p^* < p < 1$, the agent strictly hedges the option. When $p^* < p$, it can be observed that  the value of $p$ has a negligible effect on the agent's strategy.
Once the value of $p$ falls below $p^*$, however, the agent commences a strategy targeting upper tail gains by engaging in both asset trading and option hedging. This critical level $p^*$  depends on the details of the option being sold. For knock-in calls, a lower barrier makes it more likely that the option is never knocked in. It becomes easier for the agent to seek upper tail gains without the negative option payoff, and thus results in a higher $p^*$. The reverse is true for knock-out calls.
\begin{figure}
     \centering
     \begin{subfigure}{0.45\textwidth}
         \centering
         \includegraphics[width=\textwidth]{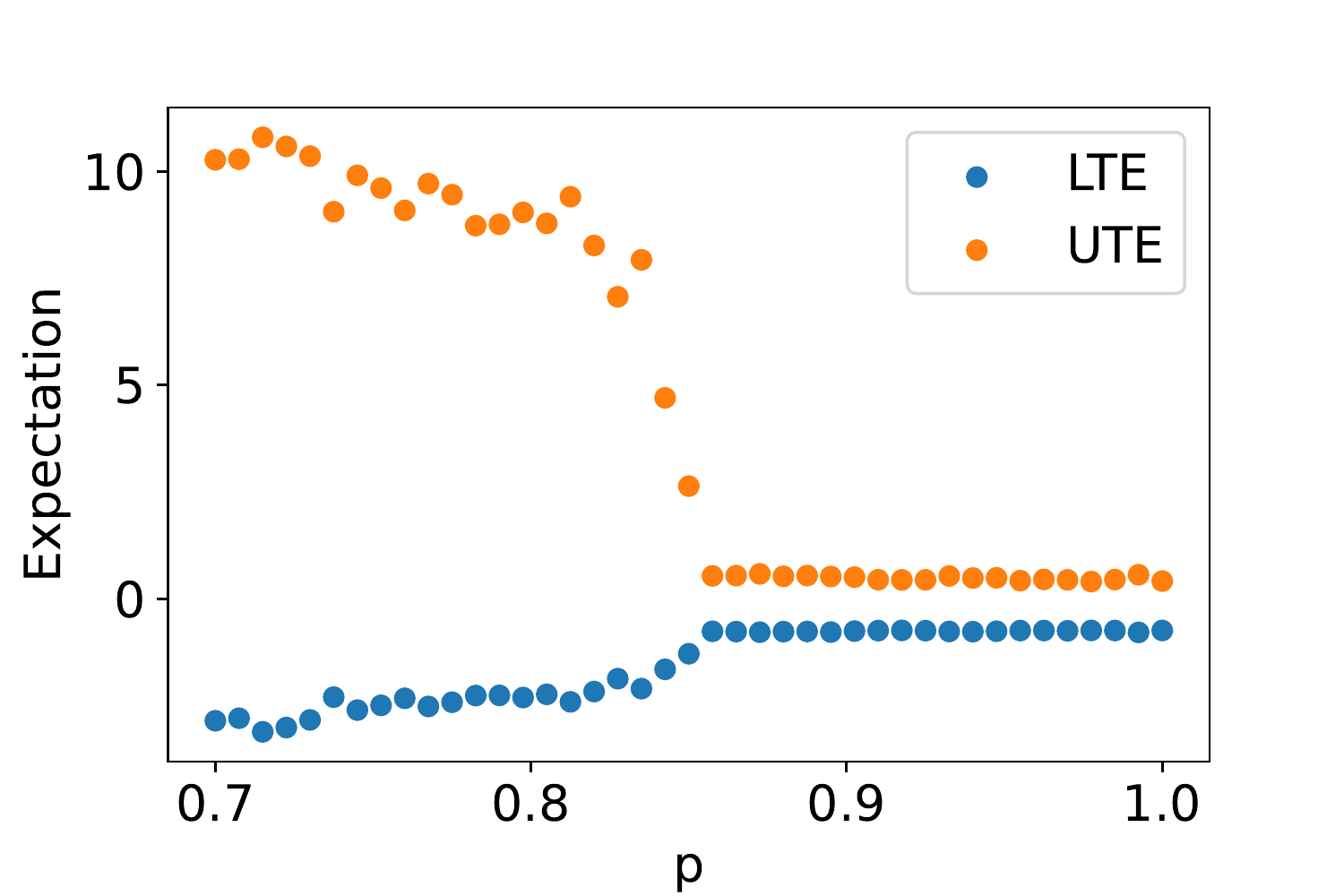}
         \caption{Knock-in Option}
     \end{subfigure}
     \hspace{0.5cm}
     \begin{subfigure}{0.45\textwidth}
         \centering
         \includegraphics[width=\textwidth]{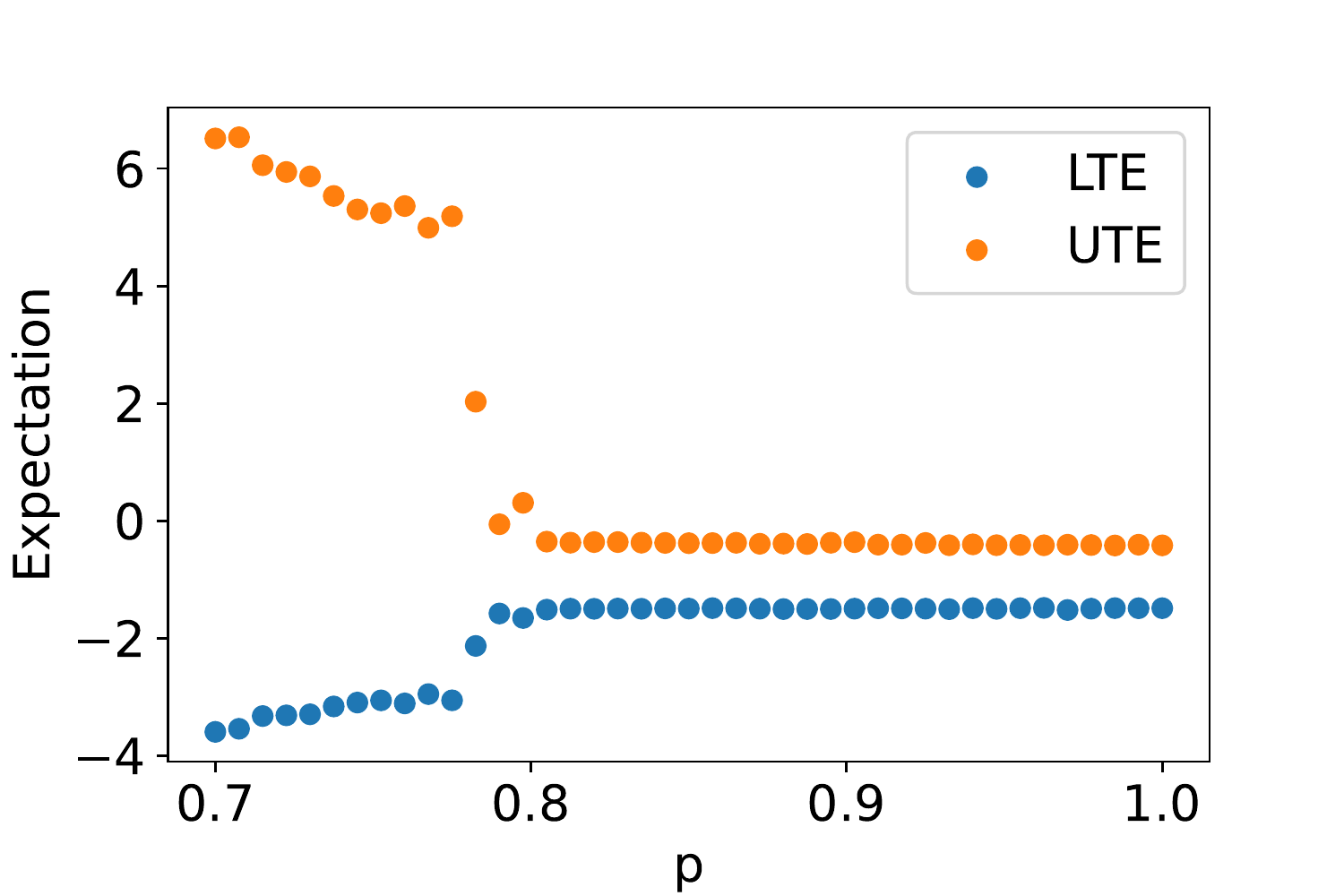}
         \caption{Knock-out option}
     \end{subfigure}
        \caption{Comparison of lower tail expectation (LTE) and upper tail expectation (UTE) for $\alpha = 0.1$ and $\beta = 0.9$ while varying $p$.}
        \label{fig:phase transition}
\end{figure}

To provide additional insight into this phenomenon, we examine the optimal actions of the agent under $\alpha-\beta$ risk measures in a scenario where there is no option exposure and inventory is subject to the constraint of being within the range of $[-M, M]$. In this case, the agent's terminal wealth $X^\phi$ is given by setting $V(0) = V(T) = 0$ in \eqref{eqn:wealth}. If $\inf_\phi R^{\alpha \beta}_p[X^\phi] = 0$, then $\Delta_i = 0$, for all $i \in \{0,..., N-1\}$, is an optimal solution -- irrespective of $p$. Not trading is an admissible strategy that attains $R^{\alpha \beta}_p[X^\phi]=0$.  If, however, there exists a strategy $\bar\Delta$ (with parameters $\bar\phi$) such that $R^{\alpha \beta}_p [X^{\bar \phi}] < 0$, then, as $R^{\alpha \beta}_p$ is positive homogeneous, the strategy $\tau\Delta_i$ (for $\tau>0$) always results in lower risk. Therefore, the strategy is to either trade nothing or to leverage fully. The latter case leads to the large increases in the UTE and this is what results in the phase transition.

As we have observed, this phenomenon seems to apply to options hedging as well. For instance, given that knock-in options behave similarly to assets when the barrier is not breached, the agent employs a strategy akin to asset hedging. The size of the jump in LTE depends on the inventory constraints. Particularly, if there are no constraints on inventory, no optimal solution exists when $p < p^*$. Furthermore, $p^*$ depends on the remaining model assumptions. For example, transaction costs decrease $p^*$ because transaction costs decrease the potential gains. Robustification also decreases $p^*$. This is due to the agent's higher risk aversion, as they aim to optimise the worst case scenario of $X^\phi$ rather than $X^\phi$ itself. The result of this effect is demonstrated by the dramatically different behaviours seen in Figure \ref{fig: ab out non robust} and Figure \ref{fig: ab out robust}. 

\section{Conclusion}

Here, we extend the results of robust risk-aware optimisation in Jaimungal et al. \cite{jaimungal2022robust} to exotic barrier options. To contextualise the strategies, we choose to focus on  barrier knock-in and knock-out call options. The methodology, however, is applicable to other options and even over the counter derivatives. This approach appears to produce results that are explainable and in line with intuition. Furthermore, robustification provides protection against uncertainty that can stem from many factors including misspecification in asset prices or changes in transaction dynamics. One issue worth noting is that while the strategies use dynamic decision making, the RDEU is not a time consistent risk measure, and the resulting optimisers should be viewed as precomit strategies. Although they are precomit strategies, they are not static in time or state. Nonetheless, it would be interesting to compare how the strategies for option hedging change when considering time-consistent dynamic risk measures using the methods developed in \cite{coache2021reinforcement}, \cite{coache2022conditionally}, and \cite{cheng2022markov}.

\section*{Acknowledgments}
This work was partially supported by the University of Toronto Data Sciences Institute.

\bibliographystyle{unsrt}  
\bibliography{references}

\end{document}

%% file: preamble.tex
\definecolor{mgreen}{rgb}{0,0.455,0.247}
\definecolor{mblue}{rgb}{0.098,0.18,0.357}
\definecolor{mred}{rgb}{0.902,0.4157,0.0196}
\definecolor{mgrey}{rgb}{0.90196,0.90,0.90}
\definecolor{mblack}{rgb}{0,0,0}

\newcommand{\F}{{\mathcal{F}}}
\newcommand{\X}{{\mathcal{X}}}

\renewcommand{\P}{{\mathbb{P}}}
\newcommand{\E}{{\mathbb{E}}}

\newcommand{\Xsp}{{\mathfrak{X}}}

\newcommand{\R}{{\mathcal{R}}}